\documentclass[a4paper,12pt]{amsart}

\usepackage{graphics}

\theoremstyle{plain}
\newtheorem{Theorem}{Theorem}
\newtheorem{Proposition}[Theorem]{Proposition}

\newtheorem{Lemma}[Theorem]{Lemma}
\newtheorem{Conjecture}[Theorem]{Conjecture}
\theoremstyle{definition}
\newtheorem{Definition}[Theorem]{Definition}

\theoremstyle{definition}
\newtheorem{definition}{Definition}

\newtheorem*{example}{Example} 

\newcommand{\bbR}{\mathbb{R}}
\newcommand{\bbZ}{\mathbb{Z}}

\newcommand{\fF}{{\mathfrak{F}}}

\newcommand{\fL}{{\mathfrak{L}}}
\newcommand{\cD}{{\mathcal{D}}}

\newcommand{\cK}{{\mathcal{K}}}

\newcommand{\End}{\mathrm{End}}
\newcommand{\cF}{{\mathcal{F}}}

\newcommand{\Hom}{\mathrm{Hom}}

\newcommand{\HH}{\mathrm{H}}
\newcommand{\Vect}{\mathrm{Vect}}

\newcommand{\half}{\frac{1}{2}}

\def\a{\alpha}

\def\l{\lambda}

 \title{Deforming the Lie Superalgebra $\mathcal{K}(1)$-Modules Of Symbols.}

 \author[A. F.]{ Faouzi AMMAR}

 \address[A. F.]{\newline
 Faouzi AMMAR\newline
 Facult\'e des Sciences, Universit\'e de Sfax, B.P.802, Sfax,
Tunisie.
 }

 \author[K. K.]{ Kaouthar KAMOUN}

 \address[K. K.]{\newline
 Kaouthar KAMOUN\newline
 Facult\'e des Sciences, Universit\'e de Sfax, B.P.802, Sfax,
Tunisie.
 }

\begin{document}

\maketitle\begin{abstract} We study non-trivial deformations of the
natural action of the Lie superalgebra $\mathcal{K}(1)$ of contact
vector fields on the (1,1)-dimensional superspace $\mathbb{R}^{1|1}$
on the space of symbols $\widetilde{{
\mathcal{S}}}_\delta^n=\bigoplus_{k=0}^n{\mathfrak{F}}_{\delta-\frac{k}{2}}$.
We calculate obstructions for integrability of infinitesimal
multi-parameter deformations and determine the complete local
commutative  algebra corresponding to the miniversal deformation.
\end{abstract}

\section{Introduction}

We consider the superspace $\mathbb{R}^{1|1}$ equipped with the
contact 1-form $\alpha=dx+\theta d\theta$ where $\theta$ is the odd
variable, the Lie superalgebra ${\mathcal{K}}(1)$ of contact
polynomial vector fields on $\mathbb{R}^{1|1}$(also called
superconformal Lie algebra see \cite{VicVan88}) and  the
${\mathcal{K}}(1)$-module of  symbols $\widetilde{{
\mathcal{S}}}_\delta^n=\bigoplus_{k=0}^n{\mathfrak{F}}_{\delta-\frac{k}{2}}$,
where ${\mathfrak{F}}_{\delta-\frac{k}{2}}$ is the module of the
weighted densities on $\mathbb{R}^{1|1}$. As Lie superalgebra
${\mathcal{K}}(1)$ is rigid as well as the Lie algebra of Virasoro
\cite{Roger07}, so one tries deformations of its modules. We will
use the framework of Fialowski ( \cite{Fialowski86} and
\cite{Fialowski88}) and Fialowski-Fuchs \cite{Fialowski99} (see also
\cite{Fuchs86}) and consider (multi-parameter) deformations over
complete local commutative algebras related to this deformation. The
first step of any approach to the deformation theory consists in the
determination of infinitesimal deformations. According to
Nijenhuis-Richardson \cite{NijenRichar67}, infinitesimal
deformations of the action of a Lie algebra on some module are
classified by the first cohomology space of the Lie algebra with
values in the module of endomorphisms of that module. In our case:
\begin{equation*}
\mathrm{H}^1_{\rm diff}\left({\mathcal{K}}(1);\End_{\rm
diff}(\widetilde{
\mathcal{S}}^n_{\delta})\right)=\oplus_{\lambda,k}\mathrm{H}^1_{\rm
diff}\left({\mathcal{K}}(1);\frak{D}_{\lambda,\lambda+k}\right),
\end{equation*}where  $\frak{D}_{\lambda,\mu}$ is the superspace of linear differential
operators from the superspace of weighted densities
${\mathfrak{F}}_{\l}$ to ${\mathfrak{F}}_{\mu}$, and hereafter
$2(\delta-\lambda), 2(\delta-\mu)\in\left\{0,\,1,\,\dots,\,n\right\}
$.

While the obstructions for integrability of this infinitesimal
deformations belong to the second cohomology space
\begin{equation*}
\mathrm{H}^2_{\rm diff}\left({\mathcal{K}}(1);\End_{\rm
diff}(\widetilde{
\mathcal{S}}^n_{\delta})\right)=\oplus_{\lambda,k}\mathrm{H}^2_{\rm
diff}\left({\mathcal{K}}(1);\frak{D}_{\lambda,\lambda+k}\right).
\end{equation*}
The odd first space $\mathrm{H}^1_{\rm
diff}\left({\mathcal{K}}(1);\frak{D}_{\lambda,\lambda+k}\right)_1$
was calculated in \cite{BasBK}: our task therefore, is to calculate
the even first space $\mathrm{H}^1_{\rm
diff}\left({\mathcal{K}}(1);\frak{D}_{\lambda,\lambda+k}\right)_0$
and the obstructions. We will prove that all the multi-parameter
deformations of the action are in fact of degree 1 or 2 in the
parameters  of deformation.

We shall give concrete explicit examples of the deformed action.

\section{Definitions and Notations}

\subsection{The Lie superalgebra of contact vector fields on
$\mathbb{R}^{1|1}$}

 Let $\mathbb{R}^{1\mid 1}$ be the superspace with
coordinates $(x,~\theta),$ where $\theta$ is the odd variables
$(\theta^2=0)$. We consider the superspace
$\mathbb{R}^{1|1}[x,\theta]$ of superpolynomial functions on
$\mathbb{R}^{1|1}$.
\begin{equation*}
\mathbb{R}^{1|1}[x,\theta]=\left\{F=f_0+f_1\theta:~f_0\hbox{ and
}f_1 \hbox{ are in } \mathbb{R}[x]\right\}
\end{equation*} where $\mathbb{R}[x]$ is the space of polynomial
functions on $\mathbb{R}$. The superspace
$\mathbb{R}^{1|1}[x,\theta]$ has a structure of superalgebra given
by the contact bracket
\begin{equation}
\{F,G\}=FG'-F'G+\frac{1}{2}(-1)^{p(F)+1}\overline{\eta}(F)\cdot
\overline{\eta}(G),
\end{equation}where $\eta=\frac{\partial}{\partial
{\theta}}+\theta\frac{\partial}{\partial x}$,
$\overline{\eta}=\frac{\partial}{\partial
{\theta}}-\theta\frac{\partial}{\partial x}$ and $p(F)$ is the
parity of $F$. Remark that $\eta\circ\eta=\frac{\partial}{\partial
x}$, so $\eta$ is sometimes called "square root" of
$\frac{\partial}{\partial x}$.

Any contact structure on $\mathbb{R}^{1\mid 1}$ can be defined by
the following $1$-form:
\begin{equation*}
\a=dx+\theta d\theta.
\end{equation*}
Let $\mathrm{Vect_P}(\mathbb{R}^{1|1})$ be the superspace of
superpolynomial vector fields on ${\mathbb{R}}^{1|1}$:
\begin{equation*}\mathrm{Vect_P}(\mathbb{R}^{1|1})=\left\{F_0\partial_x
+  F_1\partial \mid
~F_i\in\mathbb{R}^{1|n}[x,\theta]\right\},\end{equation*} where
$\partial$ stands for $\frac{\partial}{\partial{\theta}}$ and
$\partial_x$ stands for $\frac{\partial}{\partial x} $, and consider
the superspace $\mathcal{K}(1)$ of contact polynomial vector fields
on ${\mathbb{R}}^{1|1}$ defined by:
\begin{equation*}\mathcal{K}(1)=\left\{v\in\mathrm{Vect_p}(\mathbb{R}^{1|1})
~:~v\a=F\a,~~\hbox{for
some}~F\in\mathbb{R}^{1|1}[x,\theta]\right\},\end{equation*} where
$v\a$ is the Lie derivative of $\a$ along the vector fields $v$. Any
contact superpolynomial vector field on ${\mathbb{R}}^{1|1}$ can be
given by the following explicit form:
\begin{equation*}
v_F=F\partial_x
+\frac{1}{2}(-1)^{p(F)+1}\overline{\eta}(F)\overline{\eta},\;\text{
\ \ where $F\in \mathbb{R}^{1|1}[x,\theta]$.}
\end{equation*}

\subsection{The space of polynomial weighted densities on $\mathbb{R}^{1|1}$}

Recall the definition of the $\Vect_\mathrm{p}(\mathbb{R})$-module
of polynomial weighted densities on $\mathbb{R},$ where
$\Vect_\mathrm{p}(\mathbb{R})$ is the Lie algebra of polynomial
vector fields on $\mathbb{R}$. Consider the $1$-parameter action of
$\Vect_\mathrm{p}(\mathbb{R})$ on the space of polynomial functions
$\bbR[x]$, given by
\begin{equation*}
L^{\lambda}_{X\partial_x}(f)= Xf'+\lambda X'f,
\end{equation*}
where $\l\in\bbR$. Denote by $\cF_\l$ the
$\Vect_\mathrm{p}(\mathbb{R})$-module structure on $\bbR[x]$ defined
by this action. Geometrically, $\cF_\l$ is the space of polynomial
weighted densities of weight $\l$ on $\mathbb{R}$, i.e.,
\begin{equation}\label{Flamda}
\cF_\l=\{f(x)(dx)^{\lambda} | f\in \bbR[x]\}.
\end{equation}
  Now, in super setting, we have an  analogous definition
of weighted densities (see \cite{BasBK}) with \ $dx$ \ replaced by
$\a=dx+\theta d\theta.$ Consider the $1$-parameter action of
$\cK(1)$ on $\mathbb{R}[x,\theta]$ given by the rule:
\begin{equation}
\label{superaction} \fL^{\lambda}_{v_F}(G)=\fL_{v_F}(G) + \lambda
F'\cdot G,
\end{equation}
where $F,\, G\in\mathbb{R}[x,\theta]$ and $F'=\partial_{x}F$ or, in
components:
\begin{equation}
\label{deriv}
\frak{L}^{\lambda}_{v_F}(G)=L^{\lambda}_{a\partial_x}(g_0)+\frac{1}{2}~bg_1
+\left(L^{\lambda+\frac{1}{2}}_{a\partial_x}(g_1)+\lambda
g_0b'+\half g'_0 b\right)\theta,
\end{equation}
where $F=a+b\theta $, $G=g_0+g_1\theta $. In particular, we have
\begin{gather*}\left\{\begin{array}{llllllll}
\frak{L}^{\lambda}_{v_a}(g_0)=L^{\lambda}_{a\partial_x}(g_0),\,\,\,\,
\frak{L}^{\lambda}_{v_a}(g_1\theta)=\theta
L^{\lambda+\frac{1}{2}}_{a\partial_x}(g_1),\,\,\\[10pt]
\frak{L}^{\lambda}_{v_{b\theta}}(g_0)=(\lambda g_0b'+\half g'_0
b)\theta\hbox{ ~~and~~
}\frak{L}^{\lambda}_{v_{b\theta}}(g_1\theta)=\half bg_1.
\end{array}\right.
\end{gather*}
We denote this $\cK(1)$-module by ${\mathfrak{F}}_{\l}$, the space
of all polynomial weighted densities on $\mathbb{R}^{1|1}$ of weight
$\l$:
\begin{equation}
\label{densities} {\mathfrak{F}}_\l=\left\{\phi=f(x,\theta)\a^{\l}
\mid f(x,\theta) \in\mathbb{R}[x,\theta]\right\}.
\end{equation}
Let  $\frak{D}_{\lambda,\mu}
:=\Hom_{\text{diff}}(\fF_{\l},\fF_{\mu})$, be the $\cK(1)$-module of
linear differntial superoperators, the $\cK(1)$-action on this
superspace is given by:
\begin{equation}\label{d-action}
\fL^{\l,\mu}_{v_F}(A)=\fL^{\mu}_{v_F}\circ A-(-1)^{p(A)p(F)} A\circ
\fL^{\l}_{v_F}.
\end{equation}
Obviously:\label{decom}
\begin{itemize}
\item[1)] The adjoint $\cK(1)-$module, is isomorphic to
${\mathfrak{F}}_{-1}.$ \item[2)] As a $\Vect_p(\mathbb{R})$-module,
${\mathfrak{F}}_{\l}\simeq\cF_\l \oplus \Pi(\cF_{\l+\half})$, where
$\cF_\l$ is  the $\Vect_p(\mathbb{R})$-module of polynomial weighted
densities  of weight $\l$ and $ \Pi$ is the functor of the change of
parity.
\end{itemize}
\begin{Proposition}
 As a $\Vect_p({\mathbb{R}})$-module, we have for the
homogeneous relative parity components:
\begin{equation*}(\frak{D}_{\lambda,\mu})_0\simeq\cD_{\lambda,\mu}
\oplus
\cD_{\lambda+\frac{1}{2},\mu+\frac{1}{2}}\\ \hbox{and }\\
(\frak{D}_{\lambda,\mu})_1\simeq\Pi(\cD_{\lambda+\frac{1}{2},\mu}
\oplus\cD_{\lambda,\mu+\frac{1}{2}}).
\end{equation*}
\end{Proposition}
\subsection{The supertransvectants: explicit formula}


\begin{definition}(see\cite{GarOvs07}) The supertransvectants are the bilinear
$\mathfrak{osp}(1|2)$-invariant maps
$$\frak{J}^{\alpha,\beta}_{k}:\frak{F}_{\alpha}\otimes\frak{F}_{\beta}\longrightarrow
\frak{F}_{\alpha+\beta+k}$$ where $k\in
\{0,\frac{1}{2},1,\frac{3}{2},\ldots\}$. These operators were
calculated in \cite{GierThei93} (see also \cite{Huang94}), let us
give their explicit formula.

One has \begin{equation}\label{supertransvectants}
   \frak{J}^{\alpha,\beta}_{k}(f,g)=\displaystyle\sum_{i+j=2k} J_{i,j}^k
   \overline{\eta}^i(f) \overline{\eta}^j(g)
\end{equation}where the numeric coefficients are
\begin{equation}\label{numeric coefficients}
J_{i,j}^k=(-1)^{\left(\left[\frac{j+1}{2}\right]+j(i+p(f))\right)}
\frac{ \left(\begin{array}{c}
[k] \\
$[$\frac{2j+1+(-1)^{2k}}{4}$]$ \\
\end{array}
\right) \left(\begin{array}{c}
2\alpha+[k-\frac{1}{2}] \\
$[$\frac{2j+1-(-1)^{2k}}{4}$]$ \\
\end{array}
\right)} {\left(
\begin{array}{cc}
2\beta+[\frac{j-1}{2}]& \\
$[$\frac{j+1}{2}$]$& \\
\end{array}
\right)}
\end{equation}where $[a]$ denotes the integer part of $a \in
\mathbb{R}$, as above, the binomial coefficients $\left(%
\begin{array}{c}
  a \\
  b \\
\end{array}%
\right) $ are well-defined if $b$ is integer. It can be cheked
directly that those operators are, indeed,
$\mathfrak{osp}(1|2)$-invariant.
\end{definition}

\subsection{The first cohomology space $\HH^1_{\rm diff}(\cK(1),\frak{D}_{\l,\mu})$}
 Let
us first recall some fundamental concepts from cohomology
theory~(\cite{Fuchs86}). Let $\frak{g}=\frak{g}_0\oplus \frak{g}_1$
be a Lie superalgebra acting on a super vector space $V=V_0\oplus
V_1$. The space $\Hom(\frak{g},\, V)$ is ($\bbZ/2\bbZ$)-graded via
\begin{equation}
\label{grade} \Hom(\frak{g}, V)_b=\displaystyle \oplus_{a\in
(\bbZ/2\bbZ)}\Hom(\frak{g}_a, V_{a+b}); \; b\in \bbZ/2\bbZ.
\end{equation}
Let \begin{equation*}\begin{array}{c}
                       Z^1(\frak{g},~V)=\left\{\gamma\in
\Hom(\frak{g},~V);\, \gamma([g,~h])= (-1)^{p(g)p(\gamma)}g\cdot
\gamma(h) \right.\\ \left. -(-1)^{p(h)(p(g)+p(\gamma))}h\cdot
\gamma(g), \, \forall g,h\in ~\frak{g}\right\}
                     \end{array}
\end{equation*}
be the space of $1$-cocycles for the Chevalley-Eilenberg
differential. According to the $\bbZ/2\bbZ$-grading (\ref{grade}),
any $1$-cocycle $\gamma\in Z^1(\frak{g}; V)$, is broken to
$(\gamma',\gamma'')\in \Hom(\frak{g}_0,\, V)\oplus
\Hom(\frak{g}_1,\, V)$.

The first cohomology space $\HH^1_{\rm
diff}(\cK(1),\frak{D}_{\lambda,\mu})$ inherits the
($\bbZ/2\bbZ$)-grading from (\ref{grade}) and it decomposes into odd
and even parts as follows:
\begin{equation*}
\HH^1_{\rm diff}(\cK(1),\frak{D}_{\lambda,\mu})=\HH^1_{\rm
diff}(\cK(1),\frak{D}_{\lambda,\mu})_0\oplus \HH^1_{\rm
diff}(\cK(1),\frak{D}_{\lambda,\mu})_1.
\end{equation*}
The odd first space $\mathrm{H}^1_{\rm
diff}\left({\mathcal{K}}(1);\frak{D}_{\lambda,\lambda+k}\right)_1$
was calculated in \cite{BasBK}, we calculate, here, the even first
space $\mathrm{H}^1_{\rm
diff}\left({\mathcal{K}}(1);\frak{D}_{\lambda,\lambda+k}\right)_0$.
\begin{Lemma}\label{sa}
The 1-cocycle $\gamma$ is a coboundary over $\cK(1)$ if and only if
$\gamma'$ is a coboundary over $\Vect_p(\mathbb{R})$.
\end{Lemma}
\begin{proof}
See \cite{BasBK}.
\end{proof}
The following theorem
 recalls the result.
\begin{Theorem}
\label{th1}~~1) The space $\HH^1_{\rm
diff}(\cK(1),\frak{D}_{\lambda,\mu})_0$ is isomorphic to the
following:
\begin{equation*}
\HH^1_{\rm diff}(\cK(1),\frak{D}_{\lambda,\mu})_{0}\simeq\left\{
\begin{array}{ll}
\bbR&\makebox{ if }~~\mu-\lambda=0, \\[2pt]\bbR&\makebox{ if }~~\mu-\lambda=2 \makebox{ for }
 \lambda\neq-1,
\\[2pt]\bbR&\makebox{ if }~~\mu-\lambda=3\makebox{ for  }~~\lambda=0 \makebox{ or
}~~\lambda= \frac{-5}{2},\\[2pt]\bbR&\makebox{ if }~~\mu-\lambda=4\makebox{ for }
~~ \lambda=\frac{-7\pm\sqrt{33}}{4}
\\[2pt]0&\makebox { otherwise. }
\end{array}
\right.
\end{equation*} The space $\HH^1_{\rm diff}(\cK(1),\frak{D}_{\lambda,\mu})_{0}$ is
generated by the cohomology classes of the $1$-cocycles:

$\bullet$ For $\lambda=\mu$ the generator can be chosen as follows:
\begin{equation*}
\gamma_{\lambda,\lambda}(v_G)(F\a^{\lambda})=G'F\a^{\lambda},
\end{equation*}
where, here and below, $F,\,G\in \mathbb{R}^{1|1}[x,\theta].$

$\bullet$ For $\mu-\lambda=2$ and $\lambda\neq-1$  the generator can
be chosen as follows:
\begin{equation*}
\gamma_{\lambda,\lambda+2}(v_G)(F\a^{\lambda})=(2\lambda \  G^3F
+3(-1)^{p(G)} \overline{\eta}(G'')\overline{\eta}(F))\a^{\lambda+2},
\end{equation*}

$\bullet$ For $\mu-\lambda=3$ and $\lambda= 0$ the generator can be
chosen as follows:
\begin{equation*}
\begin{array}{ll}
  \gamma_{0,3}(v_G)(F\a^{0})=&\left(G^4F -(-1)^{p(G)}
\overline{\eta}(G^3)\overline{\eta}(F)+G^3F''+(-1)^{p(G)}\frac{3}{2}
\overline{\eta}(G'')\overline{\eta}(F')\right)\a^{3},\\
\end{array}
\end{equation*}

$\bullet$ For $\mu-\lambda=3$ and $\lambda= \frac{-5}{2}$ the
generator can be chosen as follows:
\begin{equation*}
\begin{array}{ll}
 \gamma_{\frac{-5}{2},\frac{1}{2}}(v_G)(F\a^{\frac{-5}{2}})=&\left(G^4F
-(-1)^{p(G)}\overline{\eta}(G^3)\overline{\eta}(F)+G^3F'
 -(-1)^{p(G)}\frac{3}{8}
\overline{\eta}(G'')\overline{\eta}(F')\right) \ \a^{\frac{1}{2}},
\end{array}
\end{equation*}

$\bullet$ For $\mu-\lambda=4$ and $\lambda=
\frac{-7\pm\sqrt{33}}{4}$  the generator can be chosen as follows:
\begin{equation*}
\begin{array}{c}
  \gamma_{\lambda,\lambda+4}(v_G)(F\a^{\lambda})=\left(G^5F
+(-1)^{p(G)}\frac{5}{2\lambda}
\overline{\eta}(G^4)\overline{\eta}(F)-\frac{5}{\lambda}G^4F'\right. \\
 \left. -(-1)^{p(G)}\frac{20}{\lambda(2\lambda+1)}
\overline{\eta}(G^3)\overline{\eta}(F')\right)  \a^{\lambda+4}. \\
\end{array}
\end{equation*}

\medskip

2)The space $\HH^1_{\rm diff}(\cK(1),\frak{D}_{\lambda,\mu})_1$ is
isomorphic to the following:
\begin{equation} \label{} \HH^1_{\rm
diff}(\cK(1),\frak{D}_{\lambda,\mu})_1\simeq\left\{
\begin{array}{ll}
\mathbb{R}^2 & \hbox{ if }~~\lambda=0,~\mu=\frac{1}{2},\\[2pt]
\mathbb{R}&\hbox{ if }~~\mu=\lambda+\frac{3}{2},\\[2pt]
\mathbb{R} &\hbox{ if }~~\mu=\lambda+\frac{5}{2}~\hbox{ for all }~\lambda,\\[2pt]
0&\hbox{ otherwise. }
\end{array}
\right.
\end{equation}
The space $\HH^1_{\rm diff}(\cK(1),\frak{D}_{\lambda,\mu})_{1}$ is
generated by the cohomology classes of the $1$-cocycles:

$\bullet$ For $\lambda=0$ and $\mu=\frac{1}{2} $, the generators can
be chosen as follows:
\begin{equation*}
\begin{array}{ccc}
  \gamma_{0,\frac{1}{2}}(v_G)(F)= {\eta}(G')F\a^{\half} & \hbox{and} &
  \widetilde{\gamma}_{0,\frac{1}{2}}(v_G)(F)=
(-1)^{p(F)}{\eta}(G'F)\a^{\half}.\\
\end{array}
\end{equation*}

$\bullet$ For $\lambda=-\frac{1}{2}~~\hbox{and}~~
\mu-\lambda=\frac{3}{2}$, the generator can be chosen as follows:
\begin{equation*}
\begin{array}{ll}
\gamma_{-\frac{1}{2},1}(v_G)(F\a^{-\half})=&\Big(
\frac{3}{2}\big({\eta}(G'')+(-1)^{p(G)}{\eta}(G'')\big)F-
\frac{1}{2}\big({\eta}(G)-(-1)^{p(G)}{\eta}(G)\big)F''+(-1)^{p(F)}\\[10pt]
&\big({\eta}(G')F'+\frac{1}{2}(
G''+(-1)^{p(G)}G''){\eta}(F)\big)+(-1)^{p(G)+1}{\eta}(G'')\\[10pt]&
\big(F+(-1)^{p(F)}F\big)\Big)\a^{1}\\[10pt]
\end{array}
\end{equation*}

$\bullet$ For $\lambda\neq-\frac{1}{2}~~\hbox{and}~~
\mu-\lambda=\frac{3}{2}$, the generator can be chosen as follows:
\begin{equation*}
\gamma_{\lambda,\lambda+\frac{3}{2}}(v_G)(F\a^{\lambda})=
\left(\overline{\eta}(G'')F\right)\a^{\lambda+\frac{3}{2}}.
\end{equation*}

$\bullet$ For  $\mu-\lambda=\frac{5}{2}$, the generator can be
chosen as follows:
\begin{equation*}\begin{array}{cc}
 \gamma_{\lambda,\lambda+\frac{5}{2}}(v_G)(F\a_1^{\lambda})= &(2\lambda G^3F+
  3(-1)^{p(G)}\overline{\eta}(G'')\overline{\eta}(F))\a^{\lambda+\frac{5}{2}}. \\
\end{array}
\end{equation*}
\end{Theorem}

\begin{proof}
The odd cohomology $\HH^1_{\rm
diff}(\cK(1),\frak{D}_{\lambda,\mu})_{1}$ was calculated in
\cite{BasBK}.

Now, we are interested in the  even cohomology.  The adjoint
$\mathcal{K}(1)$-module is  $\Vect_p(\mathbb{R})$- isomorphic to
$\Vect_p(\mathbb{R})\oplus\Pi(\mathcal{F}_{-\half})$, so, the even
1-cocycle $\gamma_0$ decomposes into two components: $\gamma_0 =
(\gamma_{00},\gamma_{11})$ where
\begin{align*}\left\{\begin{array}{lllllll}
\gamma_{00}:\Vect_p(\mathbb{R})&\rightarrow&(\frak{D}_{\l,\mu})_0, \\
 \gamma_{11}:\cF_{-\half}&\rightarrow&(\frak{D}_{\l,\mu})_1,
  \end{array}\right.
\end{align*}
\newline $\bullet$ For $\lambda=\mu$, a straightest
computation shows that $\gamma_{\lambda,\lambda}$ is prolongation of
$c_{\lambda,\lambda}(X,F)=X'F$ calculated by Feigen and Fuchs in
\cite{Fialowski99}.\newline $\bullet$ For $\mu-\lambda\geq 2$.

We have $(\frak{D}_{\lambda,\mu})_0=\mathcal{D}_{\lambda,\mu}\oplus
\mathcal{D}_{\lambda+\frac{1}{2},\mu+\frac{1}{2}}$ then the
component $\gamma_{00}$ of $\gamma$ is broken on
$(\gamma_{000},\gamma_{001})$ where

\begin{align*}\left\{\begin{array}{lllllll}
       \gamma_{000}:\Vect_p(\mathbb{R})&\rightarrow&\mathcal{D}_{\l,\mu} \hbox{ and }\\
      \gamma_{001}:\Vect_p(\mathbb{R})&\rightarrow
      &\mathcal{D}_{\lambda+\frac{1}{2},\mu+\frac{1}{2}} \\
     \end{array}\right.
\end{align*}the component $\gamma_{000}$ is a differential operator
with degree $\geq 2$ then it
 vanish on $\mathfrak{sl}(2)$ thus $\gamma_{0}$ is a
 supertransvectant by the following lemma:
\begin{Lemma}\label{sl2} (\cite{BenBouj07} Lemma 3.3.) Up to coboundary, any 1-cocycle $\gamma\in
Z^1(\cK(1),\frak{D}_{\lambda,\mu})$ vanishing on $\frak{sl}(2)$ is
$\mathfrak{osp}(1|2)$-invariant. That is, if
$\gamma(X_1)=\gamma(X_x)=\gamma(X_{x^2})=0$ then the restriction of
$\gamma$ to $\frak{osp}(1|2)$ is trivial.
\end{Lemma}
As the adjoint $\cK(1)$-module is isomorphic to $\frak{F}_{-1}$, the
1-cocycle
$$\gamma:\cK(1)\longrightarrow \frak{D}_{\l,\mu}$$ can be
looked as a differential operator:$$\gamma:\frak{F}_{-1}\otimes
\frak{F}_{\lambda}\longrightarrow\frak{F}_{\mu}.$$ We consider the
supertransvectants $\frak{J}^{-1,\lambda}_k$ as it is
$k=\mu-\lambda$. If $\mu-\lambda \geq 2$, we look for those which
are non trivial 1-cocycles. In this way we can deduce
$\gamma_{\lambda,\lambda+2}$, $\gamma_{0,3}$,
$\gamma_{-\frac{1}{2},\frac{5}{2}}$ and $\gamma_{a,a+4}$ where
$a=\frac{-7\pm\sqrt{33}}{4}$.
\end{proof}
\section{Deformation Theory and Cohomology}
Deformation theory of Lie algebra homomorphisms was first considered
for one-parameter deformations \cite{ NijenRichar67, Richardson69}.
Recently, deformations of Lie (super)algebras with multi-parameters
were intensively studied ( see,  e.g., \cite{Fialowski99,AgrALO02,
AgrBBO03,Fialowski86,Fialowski88, RogOvs98, RogOvs99}). Here we give
an outline of this theory.
\subsection{Infinitesimal deformations }
Let $\rho_0 :\frak g \rightarrow{\rm End}(V)$ be an action of a Lie
superalgebra $\frak g$ on a vector superspace $V$. When studying
deformations of the $\frak g$-action $\rho_0$, one usually starts
with {\it infinitesimal} deformations:
\begin{equation}\label{infdef}
\rho=\rho_0+t\,\gamma,
\end{equation}
where $\gamma:\frak g\to{\rm End}(V)$ is a linear map and $t$ is a
formal parameter. The homomorphism condition
\begin{equation}\label{homocond}
[\rho(x),\rho(y)]=\rho([x,y]),
\end{equation}
where $x,y\in\frak g$, is satisfied in order 1 in $t$ if and only if
$\gamma$ is a 1-cocycle. That is, the map $\gamma$ satisfies
\begin{equation*}
 \gamma[x,~y] - (-1)^{p(x)p(\gamma)}[\rho_0(x),~ \gamma(y) ]
+(-1)^{p(y)(p(x)+p(\gamma))} [\rho_0(y), ~\gamma(x)]=0.
\end{equation*}
If $\dim{\mathrm{H}^1(\frak g;{\rm End}(V))}=m$, then one can choose
1-cocycles $\gamma_1,\ldots,\gamma_m$ as a basis of
$\mathrm{H}^1(\frak g;{\rm End(V)})$ and consider the following
infinitesimal deformation
\begin{equation}
\label{InDefGen2} \rho=\rho_0+\sum_{i=1}^m{}t_i\,\gamma_i,
\end{equation}
where $t_1,\ldots,t_m$ are independent formal parameters with $t_i $
and $\gamma_i$ are the same parity i.e. $p(t_i)=p(\gamma_i)$.

For the study \  of deformations of the $\mathcal{K}(1)$-action on
$\widetilde{\mathcal{S}}^n_{\delta}$, \  we must consider the space
${\mathrm H}^1_{\rm diff}(\mathcal{K}(1),\End(\widetilde{
\mathcal{S}}^n_{\delta})).$ \ Any  infinitesimal deformation of the
$\mathcal{K}(1)$-module $\widetilde{ \mathcal{S}}^n_{\delta}$ is
then of the form
\begin{equation}\label{infdef1}
\widetilde{\frak L}_{v_F}=\frak{L}_{v_F}+\frak{L}^{(1)}_{v_F},
\end{equation}
where $\frak{L}_{v_F}$ is the Lie derivative of $\widetilde{
\mathcal{S}}^n_{\delta}$ along the vector field $v_F$ defined by
(\ref{superaction}), and
\begin{equation}\label{infpart}
\begin{array}{c}
 {\frak L}_{v_F}^{(1)}= \displaystyle\sum_{\lambda}
 \displaystyle\sum_{k=0,3,4,5}t_{\lambda, \lambda+\frac{k}{2}}\,
\gamma_{\lambda,\lambda+\frac{k}{2}}(v_F) \\
+t_{0, 3}\, \gamma_{0,3}(v_F)+t_{\frac{-5}{2}, \frac{1}{2}}\,
\gamma_{\frac{-5}{2},
\frac{1}{2}}(v_F)+\displaystyle\sum_{i=1,2}t_{a_i, a_i+4}\,
\gamma_{a_i, a_i+4}(v_F)\\
+\widetilde{t}_{0, \frac{1}{2}}\, \widetilde{\gamma}_{0,
\frac{1}{2}}(v_F)+t_{0, \frac{1}{2}}\, \gamma_{0,
\frac{1}{2}}(v_F)+t_{- \frac{1}{2},0}\, \gamma_{-
\frac{1}{2},0}(v_F),\\
\end{array}
\end{equation}where $a_1=\frac{-7-\sqrt{33}}{4}$ and
$a_2=\frac{-7+\sqrt{33}}{4}$.

Let denote that we restrict our study to the deformation
(\ref{infdef1})  for generic values of $\lambda$.

\subsection{Integrability conditions }

Consider the supercommutative associative superalgebra
$\mathbb{C}[[t_1,\ldots,t_m]]$ with unity and consider the problem
of integrability of infinitesimal deformations. Starting with the
infinitesimal deformation (\ref{InDefGen2}), we look for a formal
series
\begin{equation}
\label{BigDef2} \rho= \rho_0+\sum_{i=1}^m{}t_i\,\gamma_i+
\sum_{i,j}{}t_it_j\,\rho^{(2)}_{ij}+\cdots,
\end{equation}
where the highest-order terms
$\rho^{(2)}_{ij},\rho^{(3)}_{ijk},\ldots$ are linear maps from
$\frak g$ to ${\rm End(V)}$ with $p(\rho^{(2)}_{ij})=p(t_it_j), \
p(\rho^{(3)}_{ijk})=p(t_it_jt_k),\dots$ such that the map
\begin{equation} \label{map} \rho:\frak g\to{\rm
End(V)}\otimes\mathbb{C}[[t_1,\ldots,t_m]],
\end{equation}
satisfies the homomorphism condition (\ref{homocond}) at any order
in $t_1,\ldots,t_m$.

However, quite often the above problem has no solution. Following
\cite{Fialowski88} and \cite{AgrALO02}, we must impose extra
algebraic relations on the parameters $t_1,\ldots,t_m$ in order to
get the full deformation. Let ${\mathcal{ R}}$ be an ideal in
$\mathbb{C}[[t_1,\ldots,t_m]]$ generated by some set of relations,
the quotient
\begin{equation}
\label{TrivAlg2} {\mathcal{A}}=\mathbb{C}[[t_1,\ldots,t_m]]/{
\mathcal{R}}
\end{equation}
is a supercommutative associative superalgebra with unity, and one
can speak about  deformations with base ${ \mathcal{A}}$, (see
\cite{Fialowski99} for details). The map (\ref{map}) sends $\frak g$
to ${\rm End}(V)\otimes{ \mathcal{A}}$.
\subsection{Equivalence and the first cohomology}

The notion of equivalence of deformations over commutative
associative algebras has been considered in \cite{Fialowski88}.

\begin{Definition}  Two deformations, $\rho$ and $\rho'$ with the
same base $\mathcal{A}$ are called equivalent if there exists a
formal inner automorphism $\Psi$ of the associative superalgebra
${\rm End}(V)\otimes{\mathcal{A}}$ such that
\begin{equation*}
\Psi\circ\rho=\rho'\hbox{ and } \Psi(\mathbb{I})=\mathbb{I},
\end{equation*}
where $\mathbb{I}$ is the unity of the superalgebra ${\rm
End}(V)\otimes{ \mathcal{A}}.$
\end{Definition}

 As a consequence, two infinitesimal deformations $
\rho_1=\rho_0+t\,\gamma_1, $ and $ \rho_2=\rho_0+t\,\gamma_2, $ are
equivalent if and only if $\gamma_1-\gamma_2$ is a coboundary:
\begin{equation*}(\gamma_1-\gamma_2)(x)=(-1)^{p(x)p(A_1)}[\rho_0(x),A_1]=\delta
A_1(x),
\end{equation*}
where $A_1\in{\rm End}(V)$ and $\delta$ stands for the cohomological
Chevalley-Eilenberg  coboundary for cochains on $\frak g$ with
values in $\End(V)$ (see \cite{Fuchs86, NijenRichar67}).

So, the first cohomology space $\mathrm{H}^1(\frak g;{\rm End}(V))$
determines and classifies infinitesimal deformations up to
equivalence.

\section{ Computing the second-order Maurer-Cartan equation}
Any  infinitesimal deformation of the $\mathcal{K}(1)$-module
$\widetilde{ \mathcal{S}}^n_{\delta}$ can be integrated to a formal
deformation, such deformation is then of the form
\begin{equation}\label{infdef2}
\widetilde{\frak
L}_{v_F}=\frak{L}_{v_F}+\frak{L}^{(1)}_{v_F}+\frak{L}^{(2)}_{v_F}+\cdots,
\end{equation}where \
$\frak{L}^{(2)}_{v_F}=\sum_{i,j}{}t_it_j\,\rho^{(2)}_{ij}$, \
$\frak{L}^{(3)}_{v_F}=\sum_{i,j,k}{}t_it_jt_k\,\rho^{(3)}_{ijk}$,\
\ldots.

Setting
\begin{equation*}\varphi_t = {\rho}- \rho_0,\,\,
\mathcal{L}^{(1)}=\sum_{i=1}^m{}t_i\,\gamma_i,\,\,
\mathcal{L}^{(2)}=\sum_{i,j}{}t_it_j\,\rho^{(2)}_{ij},\,\dots,
\end{equation*}
we can rewrite the relation (\ref{homocond}) in the following way:

\begin{equation}
\label{developping} [\varphi_t(G) , \rho_0(H) ] + [\rho_0(G) ,
\varphi_t(H) ] - \varphi_t([G , H]) +\sum_{i,j > 0}
\;[\mathcal{L}^{(i)}(G) ,\mathcal{L}^{(j)}(H)] = 0.
\end{equation}

The first three terms give $(\delta\varphi_t) (G,H)$. The relation
(\ref{developping}) becomes now equivalent to:
\begin{equation}
\label{maurrer cartan} \delta\varphi_t(G,H) + \sum_{i,j >
0}[\mathcal{L}^{(i)}(G) , \mathcal{L}^{(j)}(H)]= 0.
\end{equation}\vskip 0.5cm
\begin{definition}Let $\gamma_1 , \gamma_2 : \frak
g\rightarrow \End(V)$ be two arbitrary linear maps, we denote $[\![
\ , \ ]\!]$ the {\it cup-product} defined by:\begin{equation}
\label{maurrer cartan1}
\renewcommand{\arraystretch}{1.4}
\begin{array}{l}
{}[\![\gamma_1 , \gamma_2]\!] : \frak g \otimes \frak g \rightarrow
\End(V)\\ {}[\![\gamma_1 , \gamma_2]\!] (G , H) =
(-1)^{|G||\gamma_2|}\gamma_1(G)\circ \gamma_2(H) -
(-1)^{|H|(|G|+|\gamma_2|)}\gamma_1(H)\circ
\gamma_2(G)\\{}\hbox{where } |\ \ \ | \ \ \hbox{denotes the
parity}.\end{array}
\end{equation}
\end{definition}
Expanding (\ref{maurrer cartan}) in power series in $t_1,\cdots,t_m
$, we obtain the following equation for $\mathcal{L}^{(s)}$:
\begin{equation}
\label{maurrer cartank} \delta\mathcal{L}^{(s)}(G,H)+ \sum_{i+j=s}
[\mathcal{L}^{(i)}(H) ,  \mathcal{L}^{(j)}(G)] = 0.
\end{equation}
The first non-trivial relation is
\begin{equation}\label{maurrer cartan2}\delta{\mathcal{L}^{(2)}}
=-\frac{1}{2} [\![\displaystyle\sum_{\lambda}\displaystyle\sum_{j\in
\{0,3,4,5\}}
t_{\lambda,\lambda+\frac{j}{2}}\gamma_{\lambda,\lambda+\frac{j}{2}},
\displaystyle\sum_{\lambda}\displaystyle\sum_{j\in \{0,3,4,5\}}
t_{\lambda,\lambda+\frac{j}{2}}\gamma_{\lambda,\lambda+\frac{j}{2}}]\!]
\end{equation}
Therefore, it is easy to check that for any two $1$-cocycles
$\gamma_1$ and $\gamma_2 \in Z^1 (\frak g , \End(V))$, the bilinear
map $[\![\gamma_1 , \gamma_2]\!]$ is a $2$-cocycle. The first
non-trivial relation (\ref{maurrer cartan2}) is precisely the
condition for this $2$-cocycle to be a coboundary. Moreover, if one
of the $1$-cocycles $\gamma_1$ or $\gamma_2$ is a coboundary, then
$[\![\gamma_1 , \gamma_2]\!]$ is a $2$-coboundary. We therefore,
naturally deduce that the operation (\ref{maurrer cartan1}) defines
a bilinear map:
\begin{equation}
\label{cup-product} \mathrm{H}^1 (\frak g ;\End( V))\otimes
\mathrm{H}^1 (\frak g ; \End( V))\rightarrow \mathrm{H}^2 (\frak g ;
\End( V)).
\end{equation}
All the potential obstructions   are in the image of $\mathrm{H}^1
(\frak g ; \End(V))$   under the cup-product in $\mathrm{H}^2 (\frak
g ; \End(V))$.

The bilinear map (\ref{cup-product}) can be  decomposed in
homogeneous components as follows
\begin{equation}
\label{supcup} \mathrm{H}^1 (\frak g ;\End( V))_i\otimes
\mathrm{H}^1 (\frak g ; \End( V))_j\rightarrow \mathrm{H}^2 (\frak g
; \End( V))_{i+j}
\end{equation}\vskip 0.5cm
where $i,\,j\in\mathbb{Z}/2\mathbb{Z}$.

 \subsection{Cup-products of the non-trivial 1-cocycles}

  Let us consider
the 2-cocycles
\begin{equation}
\label{kaouthar}
\renewcommand{\arraystretch}{1.4}
\begin{array}{l}
{}B_{\lambda,\lambda+k}(G,H)=\displaystyle\sum_{j\in
\{0,\frac{1}{2},1,\ldots,k\}}
t_{\lambda+j,\lambda+k}t_{\lambda,\lambda+j} \ \
[\![\gamma_{\lambda+j,\lambda+k},
\gamma_{\lambda,\lambda+j}]\!](G,H),
\end{array}
\end{equation}then, it's easy to see that:
\begin{equation}\label{Bl,l+k}
{}B_{\lambda,\lambda+k} \in \mathrm{Z}^2(\mathcal{K}(1)
,\frak{D}_{\lambda,\mu}).
\end{equation}
we compute successively the 2-cocycles $B_{\lambda,\lambda+k}(G,H)$
\ for $G=g_0+\theta g_1$ and $H=h_0+\theta h_1$ two contact vectors
and $F= f_0+\theta f_1 \in {\mathfrak{F}}_\lambda$. For generic
values of $\lambda$ we have:\vskip 0.5cm $\checkmark$ For $k=0$, let
\begin{equation*}\label{B_{lambda,lambda}}
    B_{\lambda,\lambda}(G,H)=t_{\lambda,\lambda}^2[\![\gamma_{\lambda,\lambda},
\gamma_{\lambda,\lambda}]\!] :
\mathcal{K}(1)\times\mathcal{K}(1)\longrightarrow
\frak{D}_{\lambda,\lambda}
\end{equation*}

and\begin{equation*}\label{B'_{lambda,lambda}}
B_{\lambda,\lambda}(G,H)=0
\end{equation*}

$\checkmark$ For $k= \frac{3}{2}$, let
\begin{equation*}\label{B_{lambda,lambda+3/2}}
\renewcommand{\arraystretch}{1.4}
\begin{array}{l}
  {}B_{\lambda,\lambda+\frac{3}{2}}(G,H)=
(t_{\lambda,\lambda+\frac{3}{2}}
t_{\lambda,\lambda}[\![\gamma_{\lambda,\lambda+\frac{3}{2}},\gamma_{\lambda,\lambda}]\!]
+ t_{\lambda+\frac{3}{2},\lambda+\frac{3}{2}}
t_{\lambda,\lambda+\frac{3}{2}}[\![\gamma_{\lambda+\frac{3}{2}
,\lambda+\frac{3}{2}},\gamma_{\lambda,\lambda+\frac{3}{2}}]\!])(G,H):\\{}\hskip
4cm\mathcal{K}(1)\times\mathcal{K}(1)\longrightarrow
\frak{D}_{\lambda,\lambda+\frac{3}{2}},
\end{array}
\end{equation*}
\begin{equation*}\label{B'_{lambda,lambda+3/2}}
\renewcommand{\arraystretch}{1.4}
\begin{array}{l}{}B_{\lambda,\lambda+\frac{3}{2}}(G,H)(F)
=(t_{\lambda,\lambda+\frac{3}{2}}
t_{\lambda,\lambda}-t_{\lambda+\frac{3}{2},\lambda+\frac{3}{2}}
t_{\lambda,\lambda+\frac{3}{2}})\big((h_0^3g_0'-h_0'g_0^3)f_0+(g'_0h_1''-g''_1h'_0)(f_0+\theta
f_1)\\{}\hskip 6cm+\theta(g'_1h'_1)f_0\big)\end{array}
\end{equation*}

$\checkmark$ For $k=2$, let
 \begin{equation*}\label{B_{lambda,lambda+2}}
\renewcommand{\arraystretch}{1.4}
\begin{array}{l}{}B_{\lambda,\lambda+2}(G,H)=\left(t_{\lambda,\lambda+2}
t_{\lambda,\lambda}[\![
\gamma_{\lambda,\lambda+2},\gamma_{\lambda,\lambda}]\!]+t_{\lambda+2,\lambda+2}
t_{\lambda,\lambda+2}[\![
\gamma_{\lambda+2,\lambda+2},\gamma_{\lambda,\lambda+2}]\!]\right)(G,H):\\{}\hskip4cm
\mathcal{K}(1)\times\mathcal{K}(1)\longrightarrow
\frak{D}_{\lambda,\lambda+2},\end{array}
\end{equation*}

$  B_{\lambda,\lambda+2}(G,H)(F)= (t_{\lambda,\lambda+2}
t_{\lambda,\lambda}-t_{\lambda+2,\lambda+2} t_{\lambda,\lambda+2})
\big(2\lambda(h_0^3g'_0-g_0^3h'_0)f_0+2\lambda(g'_0h_1''
-g_1''h'_0)f_1+$

\hskip3.5cm $ \theta((2\lambda+7)
  (g_0^3h'_0-h_0^3g'_0)f_1+2\lambda(h_1^3g'_0-g_1^3h'_0)f_0\big)+\theta t_{\lambda,\lambda+2}
t_{\lambda,\lambda}\big(-(2\lambda+3)$

 \hskip3.4cm $(g_0^3h'_1-h_0^3g'_1)f_0
  -3(h_1''g_0''-h_0''g_1'')f_0+3
(g_1''h'_1+g'_1h_1'')f_0\big)+\theta t_{\lambda+2,\lambda+2}
t_{\lambda,\lambda+2} $

\hskip3.5cm
$\big(-2\lambda(g'_0h_1^3-h'_0g_1^3)f_0+3(g'_0h_1''-h'_0g_1'')f'_0-
3(g'_1h''_1+h'_1g_1'')f_1\big)$

 $\checkmark$ For $k=
\frac{5}{2}$, let
 \begin{equation*}\label{B_{lambda,lambda+5/2}}
\renewcommand{\arraystretch}{1.4}
\begin{array}{l}{}B_{\lambda,\lambda+\frac{5}{2}}=\left(t_{\lambda,\lambda+\frac{5}{2}}
t_{\lambda,\lambda}[\![
\gamma_{\lambda,\lambda+\frac{5}{2}},\gamma_{\lambda,\lambda}]\!]+t_{\lambda+\frac{5}{2},
\lambda+\frac{5}{2}}t_{\lambda,\lambda+\frac{5}{2}}[\![
\gamma_{\lambda+\frac{5}{2},\lambda+\frac{5}{2}},
\gamma_{\lambda,\lambda+\frac{5}{2}}]\!]\right)(G,H):\\
\hskip4cm\mathcal{K}(1)\times\mathcal{K}(1) \longrightarrow
\frak{D}_{\lambda,\lambda+\frac{5}{2}},\end{array}
\end{equation*}
$B_{\lambda,\lambda+\frac{5}{2}}(G,H)(F)=(t_{\lambda+\frac{5}{2},
\lambda+\frac{5}{2}}t_{\lambda,\lambda+\frac{5}{2}}-t_{\lambda+\frac{5}{2},
\lambda+\frac{5}{2}}t_{\lambda,\lambda+\frac{5}{2}})\big((g_0^3h'_0
-h_0^3g'_0)f_1+  3 (g_1''h_0'-g_0'h_1'')f'_0$

\hskip3.1cm $ - \theta
  \big(-4 (g_0^3h'_0-h_0^3g'_0)f_0'+ 2\lambda(g_0^4h'_0 -h_0^4g'_0)f_0 - (2\lambda+1)(g_1^3h'_0-h_1^3g'_0)f_1$

\hskip3.1cm
$+3(g_1''h'_1+g_1'h''_1)f'_0+3(g_1''h_0''-h_1''g_0'')f_1\big)\big)+t_{\lambda+\frac{5}{2},
\lambda+\frac{5}{2}}t_{\lambda,\lambda+\frac{5}{2}}\big(3(g_1''h_0''-h_1''g_1'')f_0
  $

\hskip3.1cm$+\theta\big(-4(g_0^3h_0''-g_0''h_0^3)f_0+6
g_1''h_1''f_0-(1+2\lambda)(g_1^3h'_1+h_1^3g'_1)f_0+3(g_1''h'_0-h_1''g'_0)f'_1$

\hskip3.1cm$-
4(g'_1h_0^3-h'_1g_0^3)f_0+\theta\big((g'_1h_1^3+g_1^3h'_1)f_0 -
 4(g'_1h_0^3-h'_1g_0^3)f_1\big)\big).
$

$\checkmark$ For $k=3$, let
\begin{equation*}\label{B_{lambda,lambda+3}}\renewcommand{\arraystretch}{1.4}
\begin{array}{l}{}B_{\lambda,\lambda+3}(G,H)=t_{\lambda +\frac{3}{2},\lambda+3}
t_{\lambda,\lambda+\frac{3}{2}}[\![ \gamma_{\lambda
+\frac{3}{2},\lambda+3},\gamma_{\lambda,\lambda+\frac{3}{2}}]\!](G,H):
\mathcal{K}(1)\times\mathcal{K}(1)\longrightarrow\frak{D}_{\lambda,\lambda+3},\end{array}\end{equation*}

$B_{\lambda,\lambda+3}(G,H)(F)=-2t_{\lambda +\frac{3}{2},\lambda+3}
t_{\lambda,\lambda+\frac{3}{2}}\big(g_1''h_1''f_0+\theta(
g_1''h_1''f_1-g_0^3h_1''f_0+h_0^3g_1''f_0)\big)$\vskip0.3cm

$\checkmark$ For $k=\frac{7}{2}$, let
\begin{equation*}\label{B_{lambda,lambda+7/2}}
\begin{array}{cc}
  B_{\lambda,\lambda+\frac{7}{2}}(G,H)= & (t_{\lambda
+\frac{3}{2},\lambda+\frac{7}{2}}
t_{\lambda,\lambda+\frac{3}{2}}[\![ \gamma_{\lambda
+\frac{3}{2},\lambda+\frac{7}{2}},\gamma_{\lambda,\lambda+\frac{3}{2}}]\!] \\
  {}& + t_{\lambda +2,\lambda+\frac{7}{2}} t_{\lambda,\lambda+2}[\![
\gamma_{\lambda
+2,\lambda+\frac{7}{2}},\gamma_{\lambda,\lambda+2}]\!])(G,H): \\
{}&\mathcal{K}(1)\times\mathcal{K}(1)\longrightarrow
\frak{D}_{\lambda,\lambda+4}, \\
\end{array}
\end{equation*}

$B_{\lambda,\lambda+\frac{7}{2}}(G,H)(F)=(t_{\lambda
+\frac{3}{2},\lambda+\frac{7}{2}}
t_{\lambda,\lambda+\frac{3}{2}}-t_{\lambda +2,\lambda+\frac{7}{2}}
t_{\lambda,\lambda+2})\big(6g_1''h_1''f_1-
2\lambda(g_0^3h_1''-h_0^3g_1'')f_0 + $

\hskip3.1cm$\theta\left(-6g_1''h_1''f'_0-
2(\lambda+3)(g_1''h_1^3+g_1^3h_1'')f_0-
2(\lambda+3)(g_0^3h_1''-h_0^3g''_1)f_1\right)\big).$

 $\checkmark$ For $k=4$, let
\begin{equation*}\label{B_{lambda,lambda+4}}\renewcommand{\arraystretch}{1.4}
\begin{array}{c}
  {}{}B_{\lambda,\lambda+4}(G,H)=(t_{\lambda +\frac{3}{2},\lambda+4}
t_{\lambda,\lambda+\frac{3}{2}}[\![ \gamma_{\lambda
+\frac{3}{2},\lambda+4},\gamma_{\lambda,\lambda+\frac{3}{2}}]\!]
+t_{\lambda +\frac{5}{2},\lambda+4}
t_{\lambda,\lambda+\frac{5}{2}}[\![ \gamma_{\lambda
+\frac{5}{2},\lambda+4},\gamma_{\lambda,\lambda+\frac{5 }{2}}]\!]\\
+t_{\lambda+2 ,\lambda+4} t_{\lambda,\lambda+2}[\![ \gamma_{\lambda
+2,\lambda+4},\gamma_{\lambda,\lambda+2
}]\!])(G,H):\mathcal{K}(1)\times\mathcal{K}(1)\longrightarrow
\frak{D}_{\lambda,\lambda+4},
\end{array}
\end{equation*}
$B_{\lambda,\lambda+4}(G,H)(F)=(t_{\lambda +\frac{3}{2},\lambda+4}
t_{\lambda,\lambda+\frac{3}{2}} +t_{\lambda +\frac{5}{2},\lambda+4}
t_{\lambda,\lambda+\frac{5}{2}}+\frac{1}{3}\,t_{\lambda+2
,\lambda+4} t_{\lambda,\lambda+2})\big
(-2\lambda(g_1''h_1^3+g_1^3h_1'')f_0$

\hskip3.4cm$  + \ 6g_1''h_1''f'_0+4(g_0^3h_1''-h_0^3g_1'')f_1+
\theta\left(-(2\lambda+1)( g_1''h_1^3+g_1^3h_1'')f_1+
6g_1''h_1''f'_1\right.$

\hskip3.4cm$\left.+2\lambda(g_0^4h_1''-h_0^4g_1'')f_0-7(g_0^3h_1''-h_0^3g_1'')f'_0
+ 2\lambda (g_0^3h_1^3-h_0^3g_1^3)f_0 \right)\big),$

 $\checkmark$ For $k=\frac{9}{2}$, let
\begin{equation*}\label{B_{lambda,lambda+9/2}}\renewcommand{\arraystretch}{1.4}
\begin{array}{cc}
  B_{\lambda,\lambda+\frac{9}{2}}(G,H)= & t_{\lambda+2,\lambda+\frac{9}{2}}
t_{\lambda,\lambda+2}[\![ \gamma_{\lambda+2,\lambda+\frac{9}{2}},
\gamma_{\lambda,\lambda+2}]\!] (G,H)\\
  {} &+t_{\lambda+\frac{5}{2},\lambda+\frac{9}{2}}
t_{\lambda,\lambda+\frac{5}{2}}[\![
\gamma_{\lambda+\frac{5}{2},\lambda+\frac{9}{2}},
\gamma_{\lambda,\lambda+\frac{5}{2}}]\!](G,H): \\
  {} & \mathcal{K}(1)\times\mathcal{K}(1)
\longrightarrow\frak{D}_{\lambda,\lambda+\frac{9}{2}}, \\
\end{array}
\end{equation*}
$B_{\lambda,\lambda+\frac{9}{2}}(G,H)(F)=(t_{\lambda+2,\lambda+\frac{9}{2}}
t_{\lambda,\lambda+2}-t_{\lambda+\frac{7}{2},\lambda+\frac{9}{2}}
t_{\lambda,\lambda+\frac{5}{2}})
\left(6(\lambda+2)(g_0^3h_1''-g_1''h_0^3)f'_0-4\lambda(\lambda+4)
\right.$

\hskip3.4cm$\left.(g_0^3h_1^3-g_1^3h_0^3)f_0+2\lambda(g_0^4h_1''-h_0^4g_1'')f_0
+3(2\lambda+1)(g_1^3h_1''+g_1''h_1^3)f_1 \right.$

\hskip3.4cm$\left.-18g_1''h_1''f'_1
+\theta\big(-4\lambda(\lambda+4)(h_0^4g_0^3-g_0^4h_0^3)f_0-
3(g_0^4h_1''-g_1''h_0^4)f_1+\right.$

\hskip3.4cm$\left.(12-(2\lambda+5)(2\lambda+3))(g_1^3h_0^3-h_1^3g_0^3)f_1
+3\lambda(2\lambda+7)(g_0^3h_1''-g_1''h_0^3)f'_1\right.$

\hskip3.4cm$\left.+4\lambda(2\lambda+5)g_1^3h_1^3f_0-
6(\lambda+2)(g_1^3h_1''+g_1''h_1^3)f_0-
6\lambda(h_1^4g_1''-g_1^4h_1'')f_0\right.$

\hskip3.4cm$\left.+18 g_1''h_1''f_0''\right.$

  $\checkmark$ For $k=5$, let
\begin{equation*}\label{B_{lambda,lambda+5}}\renewcommand{\arraystretch}{1.4}
\begin{array}{l}{}{}B_{\lambda,\lambda+5}(G,H)=
t_{\lambda+\frac{5}{2},\lambda+5}t_{\lambda,\lambda+\frac{5}{2}}[\![
\gamma_{\lambda+\frac{5}{2},\lambda+5},
\gamma_{\lambda,\lambda+\frac{5}{2}}]\!](G,H):
\mathcal{K}(1)\times\mathcal{K}(1)
\longrightarrow\frak{D}_{\lambda,\lambda+5},\end{array}
\end{equation*}
$B_{\lambda,\lambda+5}(G,H)(F)=t_{\lambda+\frac{5}{2},\lambda+5}
t_{\lambda,\lambda+\frac{5}{2}}
\left(-4\lambda(h_0^3g_0^4-h_0^4g_0^3)f_0-8\lambda
(g_0^3h_1^3-g_1^3h_0^3)f_1-12(h_1''g_0^4-\right.$

\hskip3.4cm$\left.g_1''h_0^4)f_1+ 6\lambda
(g_1''h_1^4+h_1''g_1^4)f_0+6(2\lambda+1)(g_1''h_1^3+h_1''g_1^3)f'_0-18g_1''h_1''f_0''\right.$

\hskip3.4cm$\left.
  + 4\lambda(2\lambda+5)g_1^3h_1^3f_0+
  \theta\left(-10(h_0^4g_0^3-h_0^3g_0^4) f_1+8\lambda
(g_0^3h_1^4-g_1^4h_0^3)f_0+\right.\right.$

\hskip3.4cm$\left.4(2\lambda+1)(g_0^3h_1^3-g_1^3h_0^3)f'_0+28
(h_1''g_0^4-g_1''h_0^4)f'_0+ 2\lambda
(g_0^4h_1^3-h_0^4g_1^3)f_0+\right.$

\hskip3.4cm$\left. 3(2\lambda+1)(g_1''h_1^4+h_1''g_1^4)f_1+
12(1+\lambda)(g_1''h_1^3+h_1''g_1^3)f'_1-18 g_1''h_1''f_1''\right.$

\hskip3.4cm$\left.+2(2\lambda+6) (2\lambda+1)g_1^3h_1^3f_1\right))$
\begin{Proposition}\label{2cocycles}

 a) Each of the 2-cocycles:
\begin{equation*}
    B_{\lambda,\lambda+\frac{3}{2}} \hbox{ for } \lambda\neq
    -\frac{1}{2};\,B_{\lambda,\lambda+2};
    \,B_{\lambda,\lambda+\frac{5}{2}}\hbox{ for } \lambda\neq -1;\,B_{\lambda,\lambda+3}
    \hbox{ and } \,B_{\lambda,\lambda+5}
\end{equation*}
   define non trivial cohomology class. Moreover, these classes are
   linearly independant.

b) Each of the 2-cocycles $B_{\lambda,\lambda+\frac{7}{2}}$,
$B_{\lambda,\lambda+4}$ and $B_{\lambda,\lambda+\frac{9}{2}}$ is a
coboundary.
\end{Proposition}

\begin{proof}

A 2-cocycles $B_{\lambda,\lambda+k}$ for $k \in
\{\frac{3}{2},2,\frac{5}{2},3,\frac{7}{2},4,\frac{9}{2} ,5\}$ is a
coboundary if and only if satisfy:
\begin{equation}
\label{kaouthar2}B_{\lambda,\lambda+k}(G,H)(F)=\delta
b_{\lambda,\lambda+k}(G,H)(F)
\end{equation}
where
\begin{equation}\label{cobord}
\renewcommand{\arraystretch}{1.4}\begin{array}{c}
 {} b_{\lambda,\lambda+k}:\mathcal{K}(1)\longrightarrow\frak{D}_{\lambda,\lambda+k} \\
  {}\hbox{and where } \\{}\delta
b_{\lambda,\lambda+k}(G,H)(F)=b_{\lambda,\lambda+k}[G,H](F)-(-1)^{|G||b_{\lambda,\lambda+k}|}\mathcal{L}_G^{\lambda,\lambda+k}\circ
(b_{\lambda,\lambda+k})(H)(F)\\{}\hskip 2.6cm
+(-1)^{|G|(|H|+|b_{\lambda,\lambda+k}|)}\mathcal{L}_H^{\lambda,\lambda+k}\circ(b_{\lambda,\lambda+k})(G)(F) \\
\end{array}
\end{equation}
\end{proof}
For  $k \in \{\frac{3}{2},2,\frac{5}{2},3,5\}$, a direct computation
shows that those $B_{\lambda,\lambda+k}$ are non trivial 2-cocycles.

For  $k \in \{\frac{7}{2},4,\frac{9}{2}\}$, remark that those
cup-products are $\frak{osp}(1|2)$-invariant then they are
supertransvectant boundaries. A simple computation shows that:

$B_{\lambda,\lambda+4}=\alpha(\lambda,t_{\lambda})\,\delta\frak{J}^{-1,\lambda}_5
$ where
$$\alpha(\lambda,t_{\lambda})=T\frac{-3(\lambda+1)(2\lambda+1)}{5(2\lambda+4)
(2\lambda^2+7\lambda+2)},$$ and $T=(t_{\lambda
+\frac{3}{2},\lambda+4} t_{\lambda,\lambda+\frac{3}{2}} +t_{\lambda
+\frac{5}{2},\lambda+4}
t_{\lambda,\lambda+\frac{5}{2}}+\frac{1}{3}t_{\lambda+2 ,\lambda+4}
t_{\lambda,\lambda+2})$\vskip 0.5cm
$B_{\lambda,\lambda+\frac{7}{2}}=\psi(\lambda,t_{\lambda})\,\delta\frak{J}^{-1,\lambda}_{\frac{9}{2}}
$ and
$$\psi(\lambda,t_{\lambda})=T'\frac{2\lambda(2\lambda+1)}{2\lambda+3},$$where $T'=(t_{\lambda
-\frac{3}{2},\lambda+\frac{7}{2}}
t_{\lambda,\lambda+\frac{3}{2}}-t_{\lambda +2,\lambda+\frac{7}{2}}
t_{\lambda,\lambda+2})$\vskip 0.5cm
$B_{\lambda,\lambda+\frac{9}{2}}=\nu(\lambda,t_{\lambda})\,\delta\frak{J}^{-1,\lambda}_{\frac{11}{2}}$
and
$$\nu(\lambda,t_{\lambda})=-T''\frac{5(\lambda+4)}{\lambda(\lambda+1)(2\lambda+1)}$$where
$T''=(t_{\lambda+2,\lambda+\frac{9}{2}}
t_{\lambda,\lambda+2}-t_{\lambda+\frac{5}{2},\lambda+\frac{9}{2}}
t_{\lambda,\lambda+\frac{5}{2}})$.

\section{ Integrability Conditions}

In this section we obtain the necessary and sufficient integrability
conditions for the infinitesimal deformation (\ref{infdef1}).

\begin{Theorem}\label{th2}: The following  conditions\newline
1) For $2(\delta-\lambda)\in\{3,\, \dots,\,n\}$ and $\lambda\neq
-\frac{1}{2}$
\begin{equation*}
t_{\lambda, \lambda+\frac{3}{2}}\,t_{\lambda,
\lambda}-t_{\lambda+\frac{3}{2}, \lambda+\frac{3}{2}}\,t_{\lambda,
\lambda+\frac{3}{2}}=0.
\end{equation*}
2) For
   $ 2(\delta-\lambda)\in\{4,\,
\dots,\,n\}$
\begin{equation*}
t_{\lambda,\lambda+2} t_{\lambda,\lambda}=t_{\lambda+2,\lambda+2}
t_{\lambda,\lambda+2}=0,
\end{equation*}
 3) For $2(\delta-\lambda)\in\{5,\,
\dots,\,n\}$ and $\lambda\neq -1$.
\begin{equation*}
t_{\lambda,
\lambda+\frac{5}{2}}t_{\lambda,\lambda}=t_{\lambda+\frac{5}{2},
\lambda+\frac{5}{2}}t_{\lambda,\lambda+\frac{5}{2}}=0,
\end{equation*}
4) For $2(\delta-\lambda)\in\{6,\, \dots,\,n\}$
\begin{equation*}
t_{\lambda +\frac{3}{2},\lambda+3}
t_{\lambda,\lambda+\frac{3}{2}}=0,
\end{equation*}
5) For $2(\delta-\lambda)\in\left\{7,\, \dots,\,n\right\}$
\begin{equation*}
\begin{array}{c}
 t_{\lambda+\frac{7}{2},\lambda+\frac{7}{2}}(t_{\lambda
+\frac{3}{2},\lambda+\frac{7}{2}}
t_{\lambda,\lambda+\frac{3}{2}}-t_{\lambda +2,\lambda+\frac{7}{2}}
t_{\lambda,\lambda+2})=0 \\
  (t_{\lambda +\frac{3}{2},\lambda+\frac{7}{2}}
t_{\lambda,\lambda+\frac{3}{2}}-t_{\lambda +2,\lambda+\frac{7}{2}}
t_{\lambda,\lambda+2})t_{\lambda,\lambda}=0 \\
\end{array}
\end{equation*}
6) For $2(\delta-\lambda)\in\left\{8,\, \dots,\,n\right\}$
\begin{equation*}
\begin{array}{c}
   t_{\lambda+4, \lambda+4}(t_{\lambda +\frac{3}{2},\lambda+4}
t_{\lambda,\lambda+\frac{3}{2}} +t_{\lambda +\frac{5}{2},\lambda+4}
t_{\lambda,\lambda+\frac{5}{2}}+\frac{1}{3}t_{\lambda+2
,\lambda+4} t_{\lambda,\lambda+2})=0, \\
  (t_{\lambda +\frac{5}{2},\lambda+4}
t_{\lambda,\lambda+\frac{5}{2}} +t_{\lambda +\frac{5}{2},\lambda+4}
t_{\lambda,\lambda+\frac{5}{2}}+\frac{1}{3}t_{\lambda+2
,\lambda+4} t_{\lambda,\lambda+2})t_{\lambda,\lambda}=0 \\
\end{array}
\end{equation*}
7) For $2(\delta-\lambda)\in\left\{9,\, \dots,\,n\right\}$
\begin{equation*}
\begin{array}{c}
t_{\lambda+\frac{9}{2},\lambda+\frac{9}{2}}(t_{\lambda+2,\lambda+\frac{9}{2}}
t_{\lambda,\lambda+2}-t_{\lambda+\frac{5}{2},\lambda+\frac{9}{2}}
t_{\lambda,\lambda+\frac{5}{2}})=0 \\
  (t_{\lambda+2,\lambda+\frac{9}{2}}
t_{\lambda,\lambda+2}-t_{\lambda+\frac{5}{2},\lambda+\frac{9}{2}}
t_{\lambda,\lambda+\frac{5}{2}})t_{\lambda,\lambda}=0 \\
\end{array}
\end{equation*}
 8) For $2(\delta-\lambda)\in\left\{10,\,
\dots,\,n\right\}$ \begin{equation*}
\begin{array}{c} t_{\lambda,
\lambda+\frac{5}{2}}\,t_{\lambda+\frac{5}{2},
\lambda+5}=0,\\
t_{\lambda+\frac{7}{2},\lambda+5}
\left(t_{\lambda+\frac{3}{2},\lambda+\frac{7}{2}}
t_{\lambda,\lambda+\frac{3}{2}}-
t_{\lambda+2,\lambda+\frac{7}{2}}t_{\lambda,\lambda+2}\right)=0,\\
\left(t_{\lambda+3,\lambda+5} t_{\lambda+\frac{3}{2},\lambda+3}-
t_{\lambda+\frac{7}{2},\lambda+5}t_{\lambda+\frac{3}{2},\lambda+\frac{7}{2}}\right)
t_{\lambda,\lambda+\frac{3}{2}} =0,
\end{array}
\end{equation*}
9) For $2(\delta-\lambda)\in\left\{11,\, \dots,\,n\right\}$
 \begin{equation*}
\begin{array}{c}
t_{\lambda+4,\lambda+\frac{11}{2}}\left(t_{\lambda+\frac{3}{2},\lambda+4}
t_{\lambda,\lambda+\frac{3}{2}}+
t_{\lambda+\frac{5}{2},\lambda+4}t_{\lambda,\lambda+\frac{5}{2}}+\frac{1}{3}
t_{\lambda+2,\lambda+4}t_{\lambda,\lambda+2}\right)=0,\\
\left(t_{\lambda+3,\lambda+\frac{11}{2}}
t_{\lambda+\frac{3}{2},\lambda+3}+
t_{\lambda+4,\lambda+\frac{11}{2}}t_{\lambda+\frac{3}{2},\lambda+4}+\frac{1}{3}
t_{\lambda+\frac{7}{2},\lambda+\frac{11}{2}}t_{\lambda+\frac{3}{2},\lambda+\frac{7}{2}}
\right)
t_{\l,\lambda+\frac{3}{2}}=0,\\
t_{\lambda+\frac{7}{2},\lambda+\frac{11}{2}}
\left(t_{\lambda+\frac{3}{2},\lambda+\frac{7}{2}}
t_{\lambda,\lambda+\frac{3}{2}}-
t_{\lambda+2,\lambda+\frac{7}{2}}t_{\lambda,\lambda+2}\right)=0,\\

\left(t_{\lambda+\frac{7}{2},\lambda+\frac{11}{2}}
t_{\lambda+2,\lambda+\frac{7}{2}}-
t_{\lambda+4,\lambda+\frac{11}{2}}t_{\lambda+2,\lambda+4}\right)t_{\lambda,\lambda+2}=0,
\end{array}
\end{equation*}
  10) For $2(\delta-\lambda)\in\left\{12,\,
\dots,\,n\right\}$\begin{equation*}
\begin{array}{c}
t_{\lambda+\frac{9}{2},\lambda+6}
\left(t_{\lambda+\frac{5}{2},\lambda+\frac{9}{2}}
t_{\lambda,\lambda+\frac{5}{2}}-
t_{\lambda+2,\lambda+\frac{9}{2}}t_{\lambda,\lambda+2}\right)=0,\\
 \left(t_{\lambda+4,\lambda+6}
t_{\lambda+\frac{3}{2},\lambda+4}-
t_{\lambda+\frac{7}{2},\lambda+6}t_{\lambda+\frac{3}{2},\lambda+\frac{7}{2}}\right)
t_{\lambda,\lambda+\frac{3}{2}}=0,\\
t_{\lambda+4,\lambda+6}\left(t_{\lambda+\frac{3}{2},\lambda+4}
t_{\lambda,\lambda+\frac{3}{2}}+
t_{\lambda+\frac{5}{2},\lambda+4}t_{\lambda,\lambda+\frac{5}{2}}+
\frac{1}{3}\,t_{\lambda+2,\lambda+4}t_{\lambda,\lambda+2}\right)=0,\\
\left(t_{\lambda+\frac{7}{2},\lambda+6}
t_{\lambda+2,\lambda+\frac{7}{2}}+
t_{\lambda+\frac{9}{2},\lambda+6}t_{\lambda+2,\lambda+\frac{9}{2}}+
\frac{1}{3}\,t_{\lambda+4,\lambda+6}t_{\lambda+2,\lambda+4}\right)t_{\lambda,\lambda+2}=0,\\
t_{\lambda+\frac{7}{2},\lambda+6}
\left(t_{\lambda+\frac{3}{2},\lambda+\frac{7}{2}}
t_{\lambda,\lambda+\frac{3}{2}}-
t_{\lambda+2,\lambda+\frac{7}{2}}t_{\lambda,\lambda+2}\right)=0,\\
 \left(t_{\lambda+4,\lambda+6}
t_{\lambda+\frac{5}{2},\lambda+4}-
t_{\lambda+\frac{9}{2},\lambda+6}t_{\lambda+\frac{5}{2},\lambda+\frac{9}{2}}\right)
t_{\lambda,\lambda+\frac{5}{2}}=0,
\end{array}
\end{equation*}
11)  For $2(\delta-\lambda)\in\left\{13,\,
\dots,\,n\right\}$\begin{equation*}
\begin{array}{c}
  t_{\lambda+4,\lambda+\frac{13}{2}}\left(t_{\lambda+\frac{3}{2},\lambda+4}
t_{\lambda,\lambda+\frac{3}{2}}+
t_{\lambda+\frac{5}{2},\lambda+4}t_{\lambda,\lambda+\frac{5}{2}}+
\frac{1}{3}\,t_{\lambda+2,\lambda+4}t_{\lambda,\lambda+2}\right)=0, \\
 \left(t_{\lambda+4,\lambda+\frac{13}{2}}
t_{\lambda+\frac{5}{2},\lambda+4}+
t_{\lambda+5,\lambda+\frac{13}{2}}t_{\lambda+\frac{5}{2},\lambda+5}+\frac{1}{3}\,
t_{\lambda+\frac{9}{2},\lambda+\frac{13}{2}}
t_{\lambda+\frac{5}{2},\lambda+\frac{9}{2}}\right) t_{\lambda,\lambda+\frac{5}{2}}=0, \\
  t_{\lambda+\frac{9}{2},\lambda+\frac{13}{2}}
\left(t_{\lambda+\frac{5}{2},\lambda+\frac{9}{2}}
t_{\lambda,\lambda+\frac{5}{2}}-
t_{\lambda+2,\lambda+\frac{9}{2}}t_{\lambda,\lambda+2}\right)=0, \\
\left(t_{\lambda+\frac{9}{2},\lambda+\frac{13}{2}}
t_{\lambda+2,\lambda+\frac{9}{2}}-
t_{\lambda+4,\lambda+\frac{13}{2}}t_{\lambda+2,\lambda+4}\right) t_{\lambda,\lambda+2}=0, \\
\end{array}
\end{equation*}
 12)  For $2(\delta-\lambda)\in\left\{14,\,
\dots,\,n\right\}$\begin{equation*}
\begin{array}{c}t_{\lambda+\frac{9}{2},\lambda+7}
\left(t_{\lambda+\frac{5}{2},\lambda+\frac{9}{2}}
t_{\lambda,\lambda+\frac{5}{2}}-
t_{\lambda+2,\lambda+\frac{9}{2}}t_{\lambda,\lambda+2}\right)=0,
\\\left(t_{\lambda+5,\lambda+7} t_{\lambda+\frac{5}{2},\lambda+5}-
t_{\lambda+\frac{9}{2},\lambda+7}t_{\lambda+\frac{5}{2},\lambda+\frac{9}{2}}\right)
t_{\lambda,\lambda+\frac{5}{2}},=0\\
\left(t_{\lambda+\frac{3}{2},\lambda+\frac{7}{2}}
t_{\lambda,\lambda+\frac{3}{2}}-
t_{\lambda+2,\lambda+\frac{7}{2}}t_{\lambda,\lambda+2}\right)\left(t_{\lambda+5,\lambda+7}
t_{\lambda+\frac{7}{2},\lambda+5}-
t_{\lambda+\frac{11}{2},\lambda+7}t_{\lambda+\frac{7}{2},\lambda+\frac{11}{2}}\right)=0
\end{array}
\end{equation*}
13) For $2(\delta-\lambda)\in\left\{15,\,
\dots,\,n\right\}$\begin{equation*}
\begin{array}{c}\left(t_{\lambda+\frac{11}{2},\lambda+\frac{15}{2}}
t_{\lambda+4,\lambda+\frac{11}{2}}-
t_{\lambda+6,\lambda+\frac{15}{2}}t_{\lambda+4,\lambda+6}\right)
\times\\\left(t_{\lambda+\frac{3}{2},\lambda+4}
t_{\lambda,\lambda+\frac{3}{2}}+
t_{\lambda+\frac{5}{2},\lambda+4}t_{\lambda,\lambda+\frac{5}{2}}+
\frac{1}{3}t_{\lambda+2,\lambda+4}t_{\lambda,\lambda+2}\right)=0,\\
\left(t_{\lambda+\frac{3}{2},\lambda+\frac{7}{2}}
t_{\lambda,\lambda+\frac{3}{2}}-
t_{\lambda+2,\lambda+\frac{7}{2}}t_{\lambda,\lambda+2}\right)\times\\\left(t_{\lambda+5,\lambda+\frac{15}{2}}
t_{\lambda+\frac{7}{2},\lambda+5}+
t_{\lambda+6,\lambda+\frac{15}{2}}t_{\lambda+\frac{7}{2},\lambda+6}
+\frac{1}{3}t_{\lambda+\frac{11}{2},\lambda+\frac{15}{2}}
t_{\lambda+\frac{7}{2},\lambda+\frac{11}{2}}\right)=0\end{array}
\end{equation*}
14) For $2(\delta-\lambda)\in\left\{16,\,
\dots,\,n\right\}$\begin{equation*}
\begin{array}{c}\left(t_{\lambda+\frac{11}{2},\lambda+8}
t_{\lambda+4,\lambda+\frac{11}{2}}+
t_{\lambda+\frac{13}{2},\lambda+8}
t_{\lambda+4,\lambda+\frac{13}{2}}+\frac{1}{3}t_{\lambda+6,\lambda+8}
t_{\lambda+4,\lambda+6}\right)\times\\\left(t_{\lambda+\frac{3}{2},\lambda+4}
t_{\lambda,\lambda+\frac{3}{2}}+
t_{\lambda+\frac{5}{2},\lambda+4}t_{\lambda,\lambda+\frac{5}{2}}+ \frac{1}{3}t_{\lambda+2,\lambda+4}t_{\lambda,\lambda+2}\right)=0,\\
\left(t_{\lambda+6,\lambda+8} t_{\lambda+\frac{9}{2},\lambda+6}-
t_{\lambda+\frac{13}{2},\lambda+8}t_{\lambda+\frac{9}{2},\lambda+\frac{13}{2}}\right)
\left(t_{\lambda+\frac{5}{2},\lambda+\frac{9}{2}}
t_{\lambda,\lambda+\frac{5}{2}}-
t_{\lambda+2,\lambda+\frac{9}{2}}t_{\lambda,\lambda+2}\right)=0
,\\\left(t_{\lambda+\frac{3}{2},\lambda+\frac{7}{2}}
t_{\lambda,\lambda+\frac{3}{2}}-
t_{\lambda+2,\lambda+\frac{7}{2}}t_{\lambda,\lambda+2}\right)\left(t_{\lambda+6,\lambda+8}
t_{\lambda+\frac{7}{2},\lambda+6}-
t_{\lambda+\frac{11}{2},\lambda+8}t_{\lambda+\frac{7}{2},\lambda+\frac{11}{2}}\right)=0,
\end{array}
\end{equation*}
15) For $2(\delta-\lambda)\in\left\{17,\,
\dots,\,n\right\}$\begin{equation*}
\begin{array}{c}\left(t_{\lambda+\frac{3}{2},\lambda+4}
t_{\lambda,\lambda+\frac{3}{2}}+
t_{\lambda+\frac{5}{2},\lambda+4}t_{\lambda,\lambda+\frac{5}{2}}+\frac{1}{3}t_{\lambda+2,\lambda+4}t_{\lambda,\lambda+2}\right)
\times\\ \left(t_{\lambda+\frac{13}{2},\lambda+\frac{17}{2}}
t_{\lambda+4,\lambda+\frac{13}{2}}-
t_{\lambda+6,\lambda+\frac{17}{2}}t_{\lambda+4,\lambda+6}\right)=0,\\
\left(t_{\lambda+6,\lambda+\frac{17}{2}}
t_{\lambda+\frac{9}{2},\lambda+6}+
t_{\lambda+7,\lambda+\frac{17}{2}}t_{\lambda+\frac{9}{2},\lambda+7}+
\frac{1}{3}\,t_{\lambda+\frac{13}{2},\lambda+\frac{17}{2}}
t_{\lambda+\frac{9}{2},\lambda+\frac{13}{2}}\right)\times \\
\left(t_{\lambda+\frac{5}{2},\lambda+\frac{9}{2}}
t_{\lambda,\lambda+\frac{5}{2}}-
t_{\lambda+2,\lambda+\frac{9}{2}}t_{\lambda,\lambda+2}\right)=0
\end{array}
\end{equation*}
16) For $2(\delta-\lambda)\in\left\{18,\,
\dots,\,n\right\}$\begin{equation*}
\begin{array}{c}\left(t_{\lambda+7,\lambda+9}
t_{\lambda+\frac{9}{2},\lambda+7}-
t_{\lambda+\frac{13}{2},\lambda+9}t_{\lambda+\frac{9}{2},\lambda+\frac{13}{2}}\right)
\left(t_{\lambda+\frac{5}{2},\lambda+\frac{9}{2}}
t_{\lambda,\lambda+\frac{5}{2}}-
t_{\lambda+2,\lambda+\frac{9}{2}}t_{\lambda,\lambda+2}\right)=0
\end{array}
\end{equation*}

are necessary and sufficient for integrability of the
deformation~(\ref{infdef1}).
\end{Theorem}

\begin{proofname}\label{proof}:

a) \underline{The conditions  of integrability are necessary}:

If we take account of the Proposition \ref{2cocycles}, we deduce the
integrability conditions  1), 2), 3) and 4).
 Now we must calculate the
higher integrability conditions. Assume that the infinitesimal
deformation (\ref{infdef1}) can be integrated to a formal
deformation:
\begin{equation*}
\widetilde{\frak L}_{v_F}=\frak{L}_{v_F}+{\frak
L}^{(1)}_{v_F}+{\frak L}^{(2)}_{v_F}+{\frak L}^{(3)}_{v_F}+\cdots
\end{equation*}
The homomorphism  condition:
\begin{equation*}
 [\widetilde{\frak L}_{v_F},\widetilde{\frak
L}_{v_G}]=\widetilde{\frak L}_{v_{\{F,G\}}}
\end{equation*}
gives, for the third-order terms ${\frak L}^{(3)}$ which is a
particular case of the   Maurer-Cartan equation (\ref{maurrer
cartank}):
\begin{equation}\label{L3}
\delta({\frak L}^{(3)})=-\frac{1}{2}([\![{\frak L}^{(1)},{\frak
L}^{(2)}]\!]+[\![{\frak L}^{(2)},{\frak
L}^{(1)}]\!]),\end{equation}where
\begin{equation*}{\frak L}^{(2) }=
-\left(\sum_\l\psi(\lambda,t_{\lambda})
\frak{J}_\frac{9}{2}^{-1,\lambda}+\sum_\l\alpha(\lambda,t_{\lambda})
\frak{J}_{5}^{-1,\lambda}+\sum_\l\nu(\lambda,t_{\lambda})
\frak{J}_\frac{11}{2}^{-1,\lambda}\right).\end{equation*}
 The
right hand side of  (\ref{L3}) yields the following maps:

$\checkmark$ For $k=\frac{7}{2}$, let
\begin{equation*}
\begin{array}{c}
  D_{\lambda,\lambda+\frac{7}{2}}=  t_{\lambda+\frac{7}{2},\lambda+\frac{7}{2}}
 \psi(\lambda,t_{\lambda})
[\![\gamma_{\lambda+\frac{7}{2},\lambda+\frac{7}{2}},\frak{J}_{\frac{9}{2}}^{-1,\lambda}]\!]
+\psi(\lambda,t_{\lambda})t_{\lambda,\lambda}
[\![\frak{J}_{\frac{9}{2}}^{-1,\lambda},\gamma_{\lambda,\lambda}]\!] \\
  \hskip3cm:\mathcal{K}(1)\times \mathcal{K}(1)\longrightarrow
\frak{D}_{\lambda,\lambda+\frac{7}{2}}, \\
\end{array}
\end{equation*}

$\checkmark$ For $k=4$, let
\begin{equation*}
\begin{array}{c}
  D_{\lambda,\lambda+4}=  t_{\lambda+4,\lambda+4}
 \alpha(\lambda,t_{\lambda})
[\![\gamma_{\lambda+4,\lambda+4},\frak{J}_{5}^{-1,\lambda}]\!]
+\alpha(\lambda,t_{\lambda})t_{\lambda,\lambda}
[\![\frak{J}_{5}^{-1,\lambda},\gamma_{\lambda,\lambda}]\!] \\
  \hskip3cm:\mathcal{K}(1)\times \mathcal{K}(1)\longrightarrow
\frak{D}_{\lambda,\lambda+4}, \\
\end{array}
\end{equation*}

 $\checkmark$ For $k=\frac{9}{2}$, let
\begin{equation*}
\begin{array}{c}
  D_{\lambda,\lambda+\frac{9}{2}}=  t_{\lambda+\frac{9}{2},\lambda+\frac{9}{2}}
 \nu(\lambda,t_{\lambda})
[\![\gamma_{\lambda+\frac{9}{2},\lambda+\frac{9}{2}},\frak{J}_{\frac{11}{2}}^{-1,\lambda}]\!]
+\nu(\lambda,t_{\lambda})t_{\lambda,\lambda}
[\![\frak{J}_{\frac{11}{2}}^{-1,\lambda},\gamma_{\lambda,\lambda}]\!] \\
 \hskip3cm :\mathcal{K}(1)\times \mathcal{K}(1)\longrightarrow
\frak{D}_{\lambda,\lambda+\frac{9}{2}}, \\
\end{array}
\end{equation*}
 $\checkmark$ For $k=5$, let
\begin{equation*}
\begin{array}{cc}
 D_{\lambda,\lambda+5}=  t_{\lambda+\frac{7}{2},\lambda+5}
 \psi(\lambda,t_{\lambda})
[\![\gamma_{\lambda+\frac{7}{2},\lambda+5},\frak{J}_{\frac{9}{2}}^{-1,\lambda}]\!]
+\psi(\lambda+\frac{3}{2},t_{\lambda+\frac{3}{2}})t_{\lambda,\lambda+\frac{3}{2}}
[\![\frak{J}_{\frac{9}{2}}^{-1,\lambda+\frac{3}{2}},\gamma_{\lambda,\lambda+\frac{3}{2}}]\!]
\\:\mathcal{K}(1)\times \mathcal{K}(1)\longrightarrow
\frak{D}_{\lambda,\lambda+5},
\end{array}
\end{equation*}

$\checkmark$  For $k=\frac{11}{2}$, let
\begin{equation*}
\begin{array}{cc}
  D_{\lambda,\lambda+\frac{11}{2}}=t_{\lambda+4,\lambda+\frac{11}{2}}
 \alpha(\lambda,t_{\lambda})
[\![\gamma_{\lambda+4,\lambda+\frac{11}{2}},\frak{J}_{5}^{-1,\lambda}]\!]
   +\alpha(\lambda+\frac{3}{2},t_{\lambda+\frac{3}{2}})t_{\lambda,\lambda+\frac{3}{2}}
    [\![\frak{J}_{5}^{-1,\lambda+\frac{3}{2}},
\gamma_{\lambda,\lambda+\frac{3}{2}}]\!]\\
t_{\lambda+\frac{7}{2},\lambda+\frac{11}{2}}
  \psi(\lambda,t_{\lambda})
[\![\gamma_{\lambda+\frac{7}{2},\lambda+\frac{11}{2}},\frak{J}_{\frac{9}{2}}^{-1,\lambda}]\!]
   +\psi(\lambda+2,t_{\lambda+2})t_{\lambda,\lambda+\frac{3}{2}}
   [\![\frak{J}_{\frac{9}{2}}^{-1,\lambda+2},
\gamma_{\lambda,\lambda+\frac{3}{2}}]\!] \\:
 \mathcal{K}(1)\times \mathcal{K}(1)\longrightarrow
\frak{D}_{\lambda,\lambda+\frac{11}{2}}
\end{array}
\end{equation*}

$\checkmark$ For $k=6$, let
\begin{equation*}
\begin{array}{cc}
 D_{\lambda,\lambda+6}=  t_{\lambda+\frac{7}{2},\lambda+6}
 \psi(\lambda,t_{\lambda})
[\![\gamma_{\lambda+\frac{7}{2},\lambda+6},\frak{J}_{\frac{9}{2}}^{-1,\lambda}]\!]
+\psi(\lambda+\frac{5}{2},t_{\lambda+\frac{5}{2}})t_{\lambda,\lambda+\frac{5}{2}}
[\![\frak{J}_{\frac{9}{2}}^{-1,\lambda+\frac{5}{2}},
\gamma_{\lambda,\lambda+\frac{5}{2}}]\!]\\+t_{\lambda+\frac{9}{2},\lambda+6}
 \nu(\lambda,t_{\lambda})
[\![\gamma_{\lambda+\frac{9}{2},\lambda+6},\frak{J}_{\frac{11}{2}}^{-1,\lambda}]\!]
+\nu(\lambda+\frac{3}{2},t_{\lambda+\frac{3}{2}})t_{\lambda,\lambda+\frac{3}{2}}
[\![\frak{J}_{\frac{11}{2}}^{-1,\lambda+\frac{3}{2}},
\gamma_{\lambda,\lambda+\frac{3}{2}}]\!]\\+t_{\lambda+4,\lambda+6}
 \alpha(\lambda,t_{\lambda})
[\![\gamma_{\lambda+4,\lambda+6},\frak{J}_{5}^{-1,\lambda}]\!]
+\alpha(\lambda+2,t_{\lambda+2})t_{\lambda,\lambda+2}
[\![\frak{J}_{5}^{-1,\lambda+2}, \gamma_{\lambda,\lambda+2}]\!]
\\:\mathcal{K}(1)\times \mathcal{K}(1)\longrightarrow
\frak{D}_{\lambda,\lambda+6},
\end{array}
\end{equation*}

$\checkmark$ For $ k=\frac{13}{2}$, let
\begin{equation*}
\begin{array}{cc}
  D_{\lambda,\lambda+\frac{13}{2}}=  t_{\lambda+4,\lambda+\frac{13}{2}}
  \alpha(\lambda,t_{\lambda})
[\![\gamma_{\lambda+4,\lambda+\frac{13}{2}},\frak{J}_{5}^{-1,\lambda}]\!]
  +\alpha(\lambda+\frac{5}{2},t_{\lambda+\frac{5}{2}})
   t_{\lambda,\lambda+\frac{5}{2}} [\![\frak{J}_{5}^{-1,\lambda+\frac{5}{2}},
\gamma_{\lambda,\lambda+\frac{5}{2}}]\!]\\
+t_{\lambda+\frac{9}{2},\lambda+\frac{13}{2}}
  \nu(\lambda,t_{\lambda})
[\![\gamma_{\lambda+\frac{9}{2},\lambda+\frac{13}{2}},\frak{J}_{\frac{11}{2}}^{-1,\lambda}]\!]
  +\nu(\lambda+2,t_{\lambda+2})
   t_{\lambda,\lambda+2} [\![\frak{J}_{\frac{11}{2}}^{-1,\lambda+2},
\gamma_{\lambda,\lambda+2}]\!]\\:
 \mathcal{K}(1)\times \mathcal{K}(1)\longrightarrow
\frak{D}_{\lambda,\lambda+\frac{13}{2}}
\end{array}
\end{equation*}
 A direct and elementary
computation  for these cup-products  gives the conditions 5), 6),
7), 8), 9), 10), 11) and the tow first conditions of 12) of Theorem
\ref{th2} and proves that ${\frak L}^{(3)}\equiv 0$. We must then
calculate ${\frak L}^{(4)}$:
\begin{equation}\label{L4}
\delta({\frak L}^{(4)})=-[\![{\frak L}^{(2)},{\frak
L}^{(2)}]\!].\end{equation} Equation (\ref{L4}) is in fact
equivalent to the following ones:

$\checkmark$ For $k=7$, let
\begin{equation*}
\begin{array}{c}
 \Omega_{\lambda,\lambda+7}=\psi(\lambda+\frac{7}{2},t_{\lambda+\frac{7}{2}})
 \psi(\lambda,t_{\lambda})[\![\frak{J}_{\frac{9}{2}}^{-1,\lambda+\frac{7}{2}},
\frak{J}_{\frac{9}{2}}^{-1,\lambda}]\!]:\mathcal{K}(1)\times
\mathcal{K}(1)\longrightarrow \frak{D}_{\lambda,\lambda+7}
\end{array}\end{equation*}

 $\checkmark$ For $k=\frac{15}{2}$, let
\begin{equation*}
\begin{array}{c}
   \Omega_{\lambda,\lambda+\frac{15}{2}}= \psi(\lambda+4,t_{\lambda+4})
 \alpha(\lambda,t_{\lambda})[\![\frak{J}_{\frac{9}{2}}^{-1,\lambda+4},
\frak{J}_{5}^{-1,\lambda}]\!]
  +\alpha(\lambda+\frac{7}{2},t_{\lambda+\frac{7}{2}})
 \psi(\lambda,t_{\lambda})[\![\frak{J}_{5}^{-1,\lambda\frac{7}{2}},
\frak{J}_{\frac{9}{2}}^{-1,\lambda}]\!] \\
 :\mathcal{K}(1)\times
\mathcal{K}(1)\longrightarrow
\frak{D}_{\lambda,\lambda+\frac{15}{2}} \\
\end{array}
\end{equation*}

 $\checkmark$ For $k=8$, let
\begin{equation*}
\begin{array}{c}
 \Omega_{\lambda,\lambda+8}=\alpha(\lambda+4,t_{\lambda+4})
 \alpha(\lambda,t_{\lambda})[\![\frak{J}_{5}^{-1,\lambda+4},
\frak{J}_{5}^{-1,\lambda}]\!]+\nu(\lambda+\frac{7}{2},t_{\lambda+\frac{7}{2}})
\psi(\lambda,t_{\lambda})[\![\frak{J}_{\frac{11}{2}}^{-1,\lambda+\frac{7}{2}},
\frak{J}_{\frac{9}{2}}^{-1,\lambda}]\!]\\+\psi(\lambda+\frac{9}{2},t_{\lambda+\frac{9}{2}})
\nu(\lambda,t_{\lambda})
[\![\frak{J}_{\frac{9}{2}}^{-1,\lambda+\frac{9}{2}},
\frak{J}_{\frac{11}{2}}^{-1,\lambda}]\!]:\mathcal{K}(1)\times
\mathcal{K}(1)\longrightarrow \frak{D}_{\lambda,\lambda+8}
\end{array}\end{equation*}

$\checkmark$ For $k=\frac{17}{2}$, let
\begin{equation*}
\begin{array}{c}
 \Omega_{\lambda,\lambda+\frac{17}{2}}=\nu(\lambda+4,t_{\lambda+4})
 \alpha(\lambda,t_{\lambda})[\![\frak{J}_{\frac{11}{2}}^{-1,\lambda+4},
\frak{J}_{5}^{-1,\lambda}]\!]+\alpha(\lambda+\frac{9}{2},t_{\lambda+\frac{9}{2}})
\nu(\lambda,t_{\lambda})[\![\frak{J}_{5}^{-1,\lambda+\frac{9}{2}},
\frak{J}_{\frac{11}{2}}^{-1,\lambda}]\!]\\\mathcal{K}(1)\times
\mathcal{K}(1)\longrightarrow
\frak{D}_{\lambda,\lambda+\frac{17}{2}}
\end{array}\end{equation*}

$\checkmark$ For $k=9$, let
\begin{equation*}
\begin{array}{c}
 \Omega_{\lambda,\lambda+9}=\nu(\lambda+\frac{9}{2},t_{\lambda+\frac{9}{2}})
 \nu(\lambda,t_{\lambda})[\![\frak{J}_{\frac{11}{2}}^{-1,\lambda+\frac{9}{2}},
\frak{J}_{\frac{11}{2}}^{-1,\lambda}]\!]:\mathcal{K}(1)\times
\mathcal{K}(1)\longrightarrow \frak{D}_{\lambda,\lambda+9}.
\end{array}\end{equation*}
Necessery conditions for the integrability of the infinitesimal
deformation are that the differential operators\,
$\Omega_{\lambda,\lambda+k}(G,H)$ for $k\in
\{7,\frac{15}{2},\ldots,9\}$ must be coboundary. But differential
operators\, $\Omega_{\lambda,\lambda+k}(G,H)$  are
$\frak{osp}(1|2)$-invariantes, then they must be boundaries of
supertransvectants, so they satisfy
$$\Omega_{\lambda,\lambda+k}=A_k(\lambda,t_{\lambda})
\,\delta(\frak{J}^{-1,\lambda}_{k+1}).$$

 A straightforward computation shows that
 $A_k(\lambda,t_{\lambda})$ must be zero.

\vskip.3cm b) \underline{The  conditions of integrability are
sufficient}\vskip.3cm

 The solution
$\frak{L}^{(m)}$ of the Maurer-Cartan equation is defined up to a
1-cocycle and it has been shown in~\cite{Fialowski99, AgrALO02} that
different choices of solutions of the Maurer-Cartan equation
correspond to equivalent deformations. Thus, we can always reduce
$\frak{L}^{(3)}$ and $\frak{L}^{(4)}$ to zero by equivalence. Then,
by recurrence, the terms $\frak{L}^{(m)}$, for $m\geq4$, satisfy the
equation $\delta(\frak{L}^{(m)})=0$ and can also be reduced to the
identically zero map. This completes the proof of Theorem \ref{th2}.
\end{proofname}
\section{An open problem}
It seems \ to be \ an interesting \ open \ problem \ to compute \
the \ full \ cohomology \ ring \ $\mathrm{H}^\ast_{\rm
diff}\left({\mathcal{K}}(1);\frak{D}_{\lambda,\lambda+k}\right)$.
The only complete result here concerns the first cohomology space.
Proposition \ref{2cocycles} provides a lower bound for the dimension
of the second cohomology space. We formulate
\begin{Conjecture}
The space of second cohomology of $\cK(1)$ with coefficients in the
superspace  $\frak{D}_{\lambda,\mu}$ has the following structure:
\begin{equation*}
\HH^2_{\rm diff}(\cK(1),\frak{D}_{\lambda,\mu})_\simeq\left\{
\begin{array}{ll}
\bbR&\makebox{ if }~~\mu-\lambda=\frac{3}{2}\makebox{ and
}~~\lambda\neq -\frac{1}{2},
\\[2pt]\bbR&\makebox{ if }~~\mu-\lambda=\frac{5}{2} \makebox{ and }
 \lambda\neq -1,
\\[2pt]\bbR&\makebox{ if }~~\mu-\lambda\in \{2,3,5\}\makebox{ for  all}~~\lambda
,
\\[2pt]0&\makebox { otherwise. }
\end{array}
\right.
\end{equation*}
\end{Conjecture}
\section{Examples}
 We study deformations of
$\mathcal{K}(1)$-modules $\widetilde{ \mathcal{S}}^n_{\lambda+n}$
for any $n\in\mathbb{N}$ and for arbitrary generic
$\lambda\in\mathbb{R}.$
\begin{example}\textbf{1.} Let us consider the $\mathcal{K}(1)$-modules
$\widetilde{{ \mathcal{S}}}_{\lambda}^{0}$ and $\widetilde{{
\mathcal{S}}}_{\lambda+1}^{1}$.
\end{example}
\begin{Proposition} Every deformation of $\mathcal{K}(1)$-modules
$\widetilde{{ \mathcal{S}}}_{\lambda}^{0}$ and
$\widetilde{\mathcal{S}}_{\lambda+1}^{1}$ is equivalent to
infinitesimal one.
\end{Proposition}
\begin{proof}: Let us consider the $\mathcal{K}(1)$-module
${ \widetilde{\mathcal{S}}}_{\lambda}^{0}$. Any infinitesimal
deformation is given by:\begin{equation}\label{End(S l)}
  \widetilde{\mathcal{L}}_{v_F}=\mathcal{L}_{v_F}+\mathcal{L}_{v_F}^{(1)}
\end{equation}
where $\mathcal{L}_{v_F}$ is the Lie derivative of
$\widetilde{\mathcal{S}}_{\lambda}^{0}$ along the vector field $v_F$
defined by (\ref{superaction}), and
\begin{equation}\label{LvF}
  \mathcal{L}_{v_F}^{(1)}=t_{\lambda,\lambda}\gamma_{\lambda,\lambda},
\end{equation}
\begin{equation}\label{dl2}
    \partial(\mathcal{L}_{v_F}^{(2)})=t_{\lambda,\lambda}^2[\![\gamma_{\lambda,\lambda}
    ,\gamma_{\lambda,\lambda}]\!]
\end{equation}
but, by a direct computation, we show that $
[\![\gamma_{\lambda,\lambda},\gamma_{\lambda,\lambda}]\!]=0$ for all
$\lambda$, then $\partial(\mathcal{L}_{v_F}^{(2)})=0$ and for
consequence $\mathcal{L}_{v_F}^{(2)}=0$.

Now, let us consider the $\mathcal{K}(1)$-module
$\widetilde{\mathcal{S}}_{\lambda+1}^{1}$. Any infinitesimal
deformation is given by:
\begin{equation}
  \widetilde{\mathcal{L}}_{v_F}=\mathcal{L}_{v_F}+\mathcal{L}_{v_F}^{(1)}
\end{equation}
where $\mathcal{L}_{v_F}$ is the Lie derivative of $\widetilde{{
\mathcal{S}}}_{\lambda+1}^{1}$ along the vector field $v_F$ defined
by (\ref{superaction}), and
\begin{equation}
  \mathcal{L}_{v_F}^{(1)}=
  \displaystyle\sum_{j\in \{\frac{1}{2},1\}}t_{\lambda+j,\lambda+j}
  \gamma_{\lambda+j,\lambda+j}.
\end{equation}
By the same arguments, we show in this case that
$\mathcal{L}^{(2)}=0$, then the deformation is infinitesimal.
\end{proof}

\begin{example}\textbf{2.}\label{Example2} Consider  the
$\mathcal{K}(1)$-module $\widetilde{\mathcal{S}}_{\lambda+3}^{3}$.
In this case,$$ \widetilde{\mathcal{S}}_{\lambda+3}^{3}=
\displaystyle\sum_{k=0}^3{\mathfrak{F}}_{(\lambda+3)-\frac{k}{2}}.$$
For $\lambda\neq -2$, the  deformation of this
$\mathcal{K}(1)$-module is of degree 1, given by:\begin{equation*}
  \widetilde{\mathcal{L}}_{v_F}=\mathcal{L}_{v_F}+\mathcal{L}_{v_F}^{(1)}
\end{equation*}
where $\mathcal{L}_{v_F}$ is the Lie derivative of ${
\widetilde{\mathcal{S}}}_{\lambda+3}^{3}$ along the vector field
$v_F$ defined by (\ref{superaction}), $\mathcal{L}_{v_F}^{(1)}$ is
defined as:
\begin{equation*}
\begin{array}{c}
  \mathcal{L}_{v_F}^{(1)}=\displaystyle\sum_{j\in \{\frac{3}{2},2,\frac{5}{2},3\}}
  t_{\lambda+j,\lambda+j}\ \ \gamma_{\lambda+j,\lambda+j}+
  t_{\lambda+\frac{3}{2},\lambda+3} \ \
  \gamma_{\lambda+\frac{3}{2},\lambda+3},
\end{array}
\end{equation*}
\begin{equation*}
\partial({\frak L}^{(2)})=t_{\lambda+3,\lambda+3}
t_{\lambda+\frac{3}{2}, \lambda+3}[\![\gamma_{\lambda+3,
\lambda+3},\gamma_{\lambda+\frac{3}{2}, \lambda+3}]\!]
\end{equation*} and
\begin{equation*}
  {\frak L}^{(2)}=0.
\end{equation*}The conditions of integrability are:
\begin{equation}\label{n=3}
  t_{\lambda+3,\lambda+3}
t_{\lambda+\frac{3}{2}, \lambda+3}=0
\end{equation}where $\lambda+\frac{3}{2}\neq-\frac{1}{2}$ \,i. e.
$\lambda\neq -2$.

Let, in this case (i. e.  $\lambda\neq -2$),  ${ \mathcal{A}}$ be
the supercommutative associative superalgebra defined by the
quotient of $ \mathbb{C}[\![t_{\lambda+3,
\lambda+3},t_{\lambda+\frac{3}{2}, \lambda+3}]\!]$ by the ideal
$\mathcal{R}$  generated by equation (\ref{n=3}). Then, we speak
about a deformation with base ${ \mathcal{A}}$.

For $\lambda=-2$, one has $\partial({\frak L}^{(2)})=0$ then the
deformation of this $\mathcal{K}(1)$-module is equivalent to
infinitesimal one.

\end{example}
\begin{example}\textbf{4.}\label{Example4} Consider  the
$\mathcal{K}(1)$-module $\widetilde{\mathcal{S}}_{\lambda+4}^{4}$.
In this case the deformation of this $\mathcal{K}(1)$-module  has
the form:\begin{equation*}
  \widetilde{\mathcal{L}}_{v_F}=\mathcal{L}_{v_F}+\mathcal{L}_{v_F}^{(1)}
  +\mathcal{L}_{v_F}^{(2)}
\end{equation*}where
\begin{equation*}
\begin{array}{c}
  \mathcal{L}_{v_F}^{(1)}=\displaystyle\sum_{j\in \{2,\frac{5}{2},3,\frac{7}{2},4\}}
  t_{\lambda+j,\lambda+j}\ \ \gamma_{\lambda+j,\lambda+j}+
  \displaystyle\sum_{j\in \{2,\frac{5}{2}\}}
  t_{\lambda+j,\lambda+\frac{3}{2}+j} \ \
  \gamma_{\lambda+j,\lambda+\frac{3}{2}+j} \\+
  t_{\lambda+2,\lambda+4} \ \
  \gamma_{\lambda+2,\lambda+4},
\end{array}
\end{equation*}\begin{equation*}
\begin{array}{c}
  \partial({\frak L}^{(2)})=
t_{\lambda+2,\lambda+\frac{7}{2}}t_{\lambda+2,\lambda+2}
[\![\gamma_{\lambda+2,\lambda+\frac{7}{2}},\gamma_{\lambda+2,\lambda+2}]\!]
+t_{\lambda+\frac{7}{2},\lambda+\frac{7}{2}}t_{\lambda+2,\lambda+\frac{7}{2}}
[\![\gamma_{\lambda+\frac{7}{2},\lambda+\frac{7}{2}},
\gamma_{\lambda+2,\lambda+\frac{7}{2}}]\!] \\\hskip1.5cm
 + t_{\lambda+4,\lambda+4}t_{\lambda+\frac{5}{2},\lambda+4}
[\![\gamma_{\lambda+4,\lambda+4},\gamma_{\lambda+\frac{5}{2},\lambda+4}]\!]
+t_{\lambda+\frac{5}{2},\lambda+4}t_{\lambda+\frac{5}{2},\lambda+\frac{5}{2}}
[\![\gamma_{\lambda+\frac{5}{2},\lambda+4},
\gamma_{\lambda+\frac{5}{2},\lambda+\frac{5}{2}}]\!]\\\hskip1cm
+t_{\lambda+4,\lambda+4}t_{\lambda+2,\lambda+4} [\![
\gamma_{\lambda+4,\lambda+4},\gamma_{\lambda+2,\lambda+4}]\!]
+t_{\lambda+2,\lambda+4}t_{\lambda+2,\lambda+2} [\![
\gamma_{\lambda+2,\lambda+4},\gamma_{\lambda+2,\lambda+2}]\!]\\
\end{array}
\end{equation*}

The conditions of integrability are:
\begin{equation*}\label{n=4}
\begin{array}{c}
  \hskip1cm t_{\mu,\mu+\frac{3}{2}}t_{\mu,\mu}-
  t_{\mu+\frac{3}{2},\mu+\frac{3}{2}}t_{\mu,\mu+\frac{3}{2}} =0
   \hbox{ where } \mu\in\{\lambda+2,\lambda+\frac{5}{2}\} \hbox{ and } \mu\neq -\frac{1}{2},\\
t_{\mu,\mu+2}t_{\mu,\mu}=
  t_{\mu+2,\mu+2}t_{\mu,\mu+2} =0
   \hbox{ where } \mu\in\{\lambda+2\}.\\
\end{array}
\end{equation*}
\end{example}

\textbf{Acknowledgments}

We are grateful to Claude Roger for his constant support.

\end{document}